\documentclass[12pt, draftclsnofoot, onecolumn]{IEEEtran}

\usepackage{algorithm}
\usepackage{algpseudocode}
\algrenewcommand{\algorithmiccomment}[1]{\hskip3em$\bullet$ #1} 
\usepackage{cite}

\usepackage{graphicx}
\usepackage{psfrag}
\DeclareGraphicsExtensions{.xps}
\usepackage{epstopdf} 
\usepackage{subcaption}
\usepackage[cmex10]{amsmath}
\usepackage{amssymb}
\usepackage{color}
\usepackage{amsthm}
\usepackage{cases}

\usepackage{multirow}
\def\tr{\mathrm{tr}}

\usepackage{thmtools}
\newtheoremstyle{mytheoremstyle} 
{\topsep}                    
{0pt}                        
{\itshape}                   
{}                           
{\bfseries}                  
{.}                          
{.5em}                       
{}  

\theoremstyle{mytheoremstyle}
\newtheorem{theorem}{Theorem}

\newtheorem{corollary}{Corollary}

\newtheorem{remark}{Remark}

\usepackage[dvipsnames,svgnames]{xcolor}
\usepackage[textwidth=30mm]{todonotes}

\DeclareMathOperator*{\argmin}{arg\,min}

\begin{document}

\title{Joint Unicast and Multi-group Multicast Transmission in Massive MIMO Systems}

\author{{Meysam Sadeghi, Student Member, IEEE\thanks{M. Sadeghi (meysam.sadeghi@liu.se), E. Bj\"{o}rnson (emil.bjornson@liu.se), and E. G. Larsson (erik.g.larsson@liu.se) are with Department of Electrical Engineering (ISY), Link\"{o}ping University, Link\"{o}ping, Sweden. C. Yuen (yuenchau@sutd.edu.sg) is with Singapore University of Technology and Design, Singapore. T. L. Marzetta (tom.marzetta@nyu.edu) is with the Department of Electrical and Computer Engineering, New York University, New York, USA. This work was supported in part by the Swedish Research Council (VR), the Swedish Foundation for Strategic Research (SSF), and ELLIIT.  Also, it was supported by A*Star SERC project number 142-02-00043.
\newline \indent Parts of this paper will be presented at IEEE ICASSP 2018 \cite{ICASSP18}.},
Emil Bj\"{o}rnson, Member, IEEE,
Erik~G.~Larsson, Fellow, IEEE,
Chau~Yuen, Senior Member, IEEE,
and~Thomas~L.~Marzetta, Fellow, IEEE}
	\vspace{-0.8cm}}

\maketitle

\begin{abstract}
	We study the joint unicast and multi-group multicast transmission in massive multiple-input-multiple-output (MIMO) systems. We consider a system model that accounts for channel estimation and pilot contamination, and derive achievable spectral efficiencies (SEs) for unicast and multicast user terminals (UTs), under maximum ratio transmission and zero-forcing precoding. For unicast transmission, our objective is to maximize the weighted sum SE of the unicast UTs, and for the multicast transmission, our objective is to maximize the minimum SE of the multicast UTs. These two objectives are coupled in a conflicting manner, due to their shared power resource. Therefore, we formulate a multiobjective optimization problem (MOOP) for the two conflicting objectives. We derive the Pareto boundary of the MOOP analytically. As each Pareto optimal point describes a particular efficient trade-off between the two objectives of the system, we determine the values of the system parameters (uplink training powers, downlink transmission powers, etc.) to achieve any desired Pareto optimal point. Moreover, we prove that the Pareto region is convex, hence the system should serve the unicast and multicast UTs at the same time-frequency resource. Finally, we validate our results using numerical simulations.	
\end{abstract}


\section{Introduction}
The global mobile data traffic is facing an unprecedented growth and new records are expected in the upcoming years, e.g., the monthly global mobile data traffic will exceed $50$ exabytes by $2021$ \cite{EricssonMobilityReport}. From this huge amount of data, a considerable portion is of common interest, e.g., live broadcast of sporting events, news headlines, massive software update, and popular videos. To provide efficient delivery of such data, the 3rd Generation Partnership Project (3GPP) incorporated the multimedia broadcast/multicast service (MBMS) in the third and fourth generations of cellular networks \cite{lecompte2012evolved}. This in turn has motivated a lot of research on physical layer multicasting \cite{sidiropoulos2006transmit,karipidis2008quality}.

In physical layer multicasting\footnote{Hereafter, for brevity we refer to physical layer multicasting as multicasting or multicast transmission.} we have multiple independent data streams, each of which is of interest for a group of user terminals (UTs). Each group is called a multicasting group and the purpose of multicast transmission is to utilize the channel state information (CSI) at the transmitter to optimize the transmission based on a desired performance metric \cite{sidiropoulos2006transmit,karipidis2008quality}. A common performance metric is the max-min fairness (MMF), where we want to maximize the minimum signal-to-interference-plus-noise-ratio (SINR) or spectral efficiency (SE) of the system given a limited transmit power \cite{sidiropoulos2006transmit,karipidis2008quality,christopoulos2014weighted,tran2014conic,christopoulos2015multicast,MeysamMultiComplexity,Zhengzheng2014,sadeghi2015multi,YangMulticat,GC2017,MeysamMMFTWC}. A first seminal treatment of multicast transmission is presented in \cite{sidiropoulos2006transmit}, where the MMF problem for single-group multicasting is studied. Therein it is proved that the problem is NP-hard and a semidefinite relaxation (SDR) method is presented to approximately solve the problem. The extension of this problem to multi-group multicasting is then studied in \cite{karipidis2008quality}. The multi-group multicasting under per-antenna power constrains is investigated in \cite{christopoulos2014weighted}. These works employ SDR-based algorithms to design the multicast transmission, which suffers from high computational complexity. Considering a single-cell system with an $N$ antennas base station (BS), the SDR-based methods have a complexity of $\mathcal{O}(N^{6.5})$ \cite{karipidis2008quality}. This complexity is highly prohibitive, especially if we have large dimensional systems, e.g., massive multiple-input-multiple-output (MIMO) systems.

Massive MIMO, due to its high energy and spectral efficiency \cite{marzetta2010noncooperative,hoydis2013massive,ngo2013energy,LSAPowerNorm}, is one of the key technologies for the fifth generation of cellular networks \cite{boccardi2014five}. Therefore, there has been an increasing interest in developing multicasting algorithms tailored for massive MIMO systems \cite{tran2014conic,christopoulos2015multicast,MeysamMultiComplexity,Zhengzheng2014,sadeghi2015multi,YangMulticat,GC2017,MeysamMMFTWC}. In \cite{tran2014conic}, a successive convex approximation based method is proposed to design the precoding vector for single-group massive MIMO multicasting, which is computationally less demanding than SDR-based methods. It is then extended to multi-group multicasting in \cite{christopoulos2015multicast}, where its complexity is $\mathcal{O}(N^{3.5})$. Recently, computationally efficient precoders are introduced in \cite{MeysamMultiComplexity} that resolve the complexity problem of massive MIMO multicasting. However, the proposed methods in \cite{tran2014conic,christopoulos2015multicast,MeysamMultiComplexity} require perfect CSI to be available at the BS and UTs, while imperfect CSI is the case of practical interest.

To jointly address both problems, the imperfect CSI and the high computational complexity of massive MIMO multicasting, two approaches have been presented in the literature \cite{Zhengzheng2014,sadeghi2015multi,YangMulticat,GC2017,MeysamMMFTWC}. The first approach exploits the asymptotic orthogonality of the channels in massive MIMO systems to simplify the MMF problem \cite{Zhengzheng2014,sadeghi2015multi}. However, it requires an extremely large number of antennas to provide a reasonable performance, e.g., $N>1000$ \cite{sadeghi2015multi,GC2017}. The insufficiency of the asymptotic approach is detailed in \cite{sadeghi2015multi,GC2017}. The other approach, i.e., \cite{YangMulticat,GC2017,MeysamMMFTWC}, uses the statistical properties of massive MIMO systems and a novel pilot assignment strategy (which will be explained in Section II.A) and present practical massive MIMO multicasting methods with reasonable number of BS antennas, e.g., $N > 100$ \cite{YangMulticat,MeysamMMFTWC}.

The aforementioned results are for pure multicast transmission, but a practical system should be able to simultaneously serve both unicast and multicast UTs. This has motivated the study of joint unicast and multicast transmission, e.g., \cite{Ding_noma_uni_mul,Lv_mul_cognitive,Zhang_NonOrtho_MBMS,zhang2016massiveUniMul,larsson2016joint}. In \cite{Ding_noma_uni_mul}, non-orthogonal multiple access (NOMA) empowered joint unicast and multicast transmission is studied. In \cite{Lv_mul_cognitive}, the application of NOMA for joint transmission of unicast primary user and multicast secondary users is studied. A similar approach is also used in \cite{Zhang_NonOrtho_MBMS} for MBMS transmission in heterogeneous networks, where tractable models for analyzing the performance of
NOMA-based MBMS transmission is presented. However, the joint unicast and multicast transmission problem has so far been overlooked in the literature on massive MIMO systems. The first step in this direction is presented in \cite{zhang2016massiveUniMul}, while considering perfect CSI at the BS and UTs. A more practical approach is presented in \cite{larsson2016joint}, where a joint beamforming and broadcasting technique for massive MIMO system is studied. However, \cite{larsson2016joint} neither considers multi-group multicasting nor studies the MMF of multicast UTs or the sum~SE~(SSE) of unicast UTs.

To the best of our knowledge, the coexistence of unicast and multi-group multicast transmission in massive MIMO systems, while accounting for the computational complexity, imperfect CSI, and pilot contamination, has not been studied in the literature. Studying such a system is challenging because for multicast UTs usually the desired objective is the MMF \cite{tran2014conic,christopoulos2015multicast,MeysamMultiComplexity,Zhengzheng2014,sadeghi2015multi,YangMulticat,GC2017,MeysamMMFTWC}, while for the unicast UTs it is usually their SSE that is considered \cite{SumRateMaxPowerAllTWC,LowComSumRateMassiveMIMO,DualitySumrateGaussianMIMO,SumRateTWC,SumPower,SumRateMaxTSP,sadeghiGC14austin}. Noting that we have a shared power, inter-user interference, and two conflicting objectives, one challenging aspect of such a problem is to find a rigorous definition of optimality. Moreover, for a given optimality measure, how should the resources be allocated to reach optimality? Should we set aside certain time-frequency resources for multicast transmission and the rest for unicast transmission, or should we spatially multiplex unicast and multicast transmission?

In this paper, we answer the aforementioned questions. More precisely, we consider a joint unicast and multi-group multicast single-cell massive MIMO system, while accounting for channel estimation, power control, pilot contamination, and an arbitrary path-loss. Under this system model and for the maximum ratio transmission (MRT) and zero-forcing (ZF) precoding schemes, we provide the following:
\begin{itemize}
	\item We derive closed-form expressions for the achievable SE of each unicast and multicast UT in the system.
	\item We formulate the problem of maximizing the SSE of the unicast UTs and also the MMF problem for the multicast UTs. Given a total available power at the BS and a fixed amount of power for unicast transmission, we derive the optimal value of the MMF problem, the optimal uplink pilots' powers, the optimal downlink payloads' powers, and the optimal pilot length of unicast UTs, all in closed form. Also given a fixed amount of power for multicast transmission, we derive the optimal value of SSE problem, the optimal uplink pilots' powers, the optimal downlink payloads powers, and the optimal pilot length of multicast UTs in closed form.
	\item We show that the MMF and SSE problem are coupled in a conflicting manner such that improving one would degrade the other one. Hence we formulate a multiobjective optimization problem (MOOP) for the joint unicast and multicast transmission.
	\item We solve this MOOP and derive its Pareto boundary analytically. Moreover, for any desired point on the Pareto boundary, we determine the value of each of the decision variables such that the performance of the desired Pareto boundary point is achieved.
	\item We prove that the Pareto region is convex. Hence, any point of the Pareto region can be obtained by spatial multiplexing of unicast and multicast UTs.
\end{itemize}

The rest of the paper is organized as follows. Section~II presents the system model. Section~III derives the achievable SEs for the unicast and multicast UTs in the system. Section~IV elaborates the SSE and MMF problems, their optimal solutions, and their conflicting behaviors. Section~V introduces their associated MOOP, derives its Pareto boundary, and show its attainable set is convex. Section~VI presents the numerical results. Finally, the paper is concluded in Section~VII.

\textit{Notations}: Scalars are denoted by lower case letters whereas boldface lower (upper) case letters are used for vectors (matrices). We denote by $\mathbf{I}_{N}$ the identity matrix of size $N$. The symbol $\mathcal{CN}(\mathbf{0},\mathbf{C})$ denotes the circularly symmetric complex Gaussian distribution with zero mean and variance $\mathbf{C}$. The transpose, conjugate transpose, and expectation operators are denoted by $(.)^T$ , $(.)^H$, and $\mathbb{E}[.]$, respectively.


\section{System Model}
We consider joint unicast and multi-group multicast transmission in a single-cell massive MIMO system. We assume the system has a BS with $N$ antennas and jointly serves $U$ single-antenna unicast UTs and $G$ multicasting groups, where the $g$th multicast group has $K_{g}$ single-antenna multicast UTs. Note that considering single-antenna UTs is a worst-case assumption since in multicast transmission we need to support many types of UTs, some having one and some having multiple antennas. Therefore it is relevant to consider single-antenna UTs. Further details are given in \cite[Sec. VII. E]{marzetta2010noncooperative}. We denote the set of indices of the $U$ unicast UTs as $\mathcal{U}$ and the set of indices of the $G$ multicasting groups as $\mathcal{G}$, i.e., $\mathcal{U}=\{ 1, \ldots, U\}$ and $\mathcal{G}=\{1, \ldots, G \}$. We also denote the set of indices of the $K_{g}$ multicast UTs of group $g$ as $\mathcal{K}_{g}$, i.e., $\mathcal{K}_{g} = \{ 1, \ldots, K_{g} \}$. We assume a UT in the system is either a unicast UT or a multicast UT. Therefore the total number of UTs in the system is $U + \sum_{g=1}^{G} K_{g}$ and these UTs are arbitrarily distributed in the system. Fig. \ref{Fig.sys} presents a joint unicast and multi-group multicast transmission in a single-cell massive MIMO system with $U=4$ unicast UTs and $G=3$ groups of multicast UTs, where the first, second, and third multicasting groups have $K_{1}=7$, $K_{2}=12$, $K_{3}=16$ multicast UTs.

\begin{figure}[]
	\centering
	\includegraphics[width=0.85 \columnwidth, trim={0cm 0cm 0cm 0cm},clip]{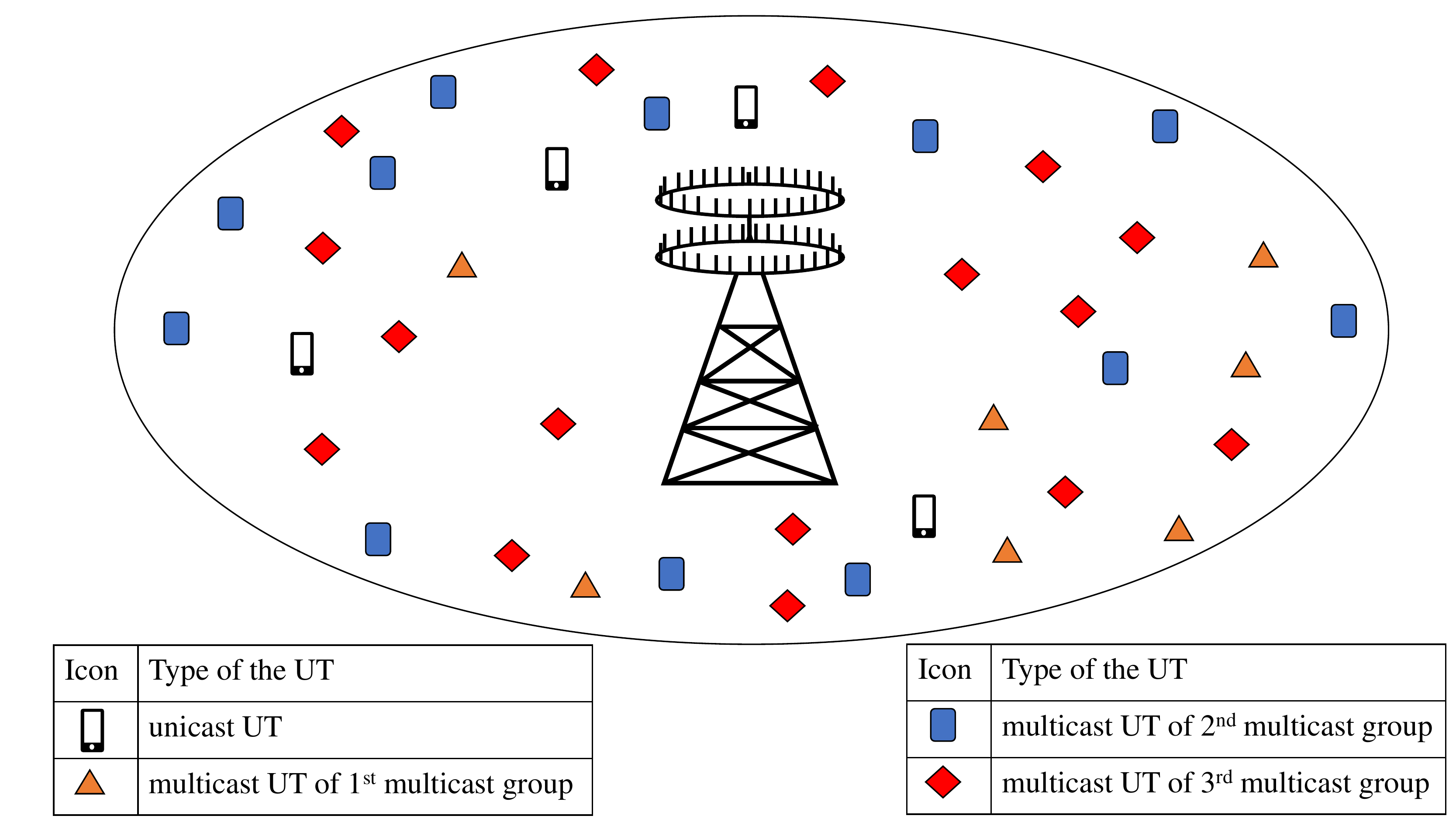}
	\caption{An example of the considered system model with $U\!\!=\!4$, $G=3$ with $K_{1}\!\!=\!7$, $K_{2}\!\!=\!12$, and $K_{3}\!\!=\!16$.}
	\label{Fig.sys}
\end{figure}

We consider a block fading channel model where the channels are assumed static within a coherence interval of $T$ symbols, where $T = C_{B} C_{T}$ with $C_{B}$ and $C_{T}$ being the coherence bandwidth and coherence time, respectively \cite{marzetta2016fundamentals}. We denote the channel response between the unicast UT $u$ and the BS as $\mathbf{f}_{u}$, and the channel response between the multicast UT $k$ in group $g$ and the BS as $\mathbf{g}_{gk}$. We consider uncorrelated Rayleigh fading channel responses for UTs, i.e., $\mathbf{f}_{u} \sim \mathcal{CN}(\mathbf{0},\beta_{u} \mathbf{I}_{N})$ and $\mathbf{g}_{gk} \sim \mathcal{CN}(\mathbf{0},\eta_{gk} \mathbf{I}_{N})$, where $\beta_{u}$ and $\eta_{gk}$ are the large-scale fading coefficients for the unicast and multicast UTs, respectively. Note that practical channels might have spatially correlated fading or line of sight components, but theoretical studies \cite{YangLOS} and practical measurements carried out in real massive MIMO propagation environments \cite{GaoMeasurements} have shown that the SE can be predicted using uncorrelated fading models. Moreover, the considered channel model enables us to present novel insights into joint unicast and multicast massive MIMO systems.

\subsection{Channel Estimation}
We assume the BS acquires CSI from uplink training using a time-division duplexing (TDD) protocol. In each coherence interval we have two phases: Uplink pilot transmission and downlink data transmission.\footnote{Note that the coherence interval has three parts: 1) pilot transmission, 2) uplink data transmission, 3) downlink data transmission. But as for multicast UTs we just have downlink data, we focus on the downlink data transmission. Note the inclusion of unicast uplink data transmission is trivial and do not change the analysis.} During uplink transmission the UTs send uplink pilots, which enables the BS to estimate the channels. Then the BS uses these channel estimates to perform downlink precoding, thanks to the reciprocity of the uplink and downlink channels. As the coherence interval has limited length, we cannot assign a dedicated orthogonal pilot to each UT in the system. In the classical unicast massive MIMO systems, it is assumed that the pilots are orthogonal within each cell and are reused in different cells. Although this approach is perfectly viable for unicast transmission, it cannot be used for multicast transmission because in multicast transmission we have a large number of UTs in each multicasting group, e.g., assume multicast transmission of TV channels where hundreds of users may watch a channel, which would exhaust the resources within each coherence interval. Therefore, for multicast UTs, we use the novel co-pilot assignment strategy proposed in \cite{YangMulticat}. This strategy assigns a shared pilot to all the UTs in each multicasting group. Therefore, the BS only requires $U+G$ pilots (rather than $U + \sum_{j=1}^{G} K_{j}$ pilots) to simultaneously serve the $U$ unicast UTs and also all the multicast UTs in the system.

We define a $\tau \times U$ pilot matrix $\mathbf{\Psi}_{un}= \sqrt{\tau} [\boldsymbol{\psi}_{un-1},\ldots,\boldsymbol{\psi}_{un-U}]$ for the unicast UTs, where $\boldsymbol{\psi}_{un-u}$ is a pilot of length $\tau$ symbols that is assigned to the $u$th unicast UT, and a $\tau \times G$ pilot matrix $\mathbf{\Psi}_{mu}= \sqrt{\tau} [\boldsymbol{\psi}_{mu-1},\ldots,\boldsymbol{\psi}_{mu-G}]$ for the multicast UTs, where $\boldsymbol{\psi}_{mu-g}$ is a pilot of length $\tau$ symbols that is assigned to every multicast UT $k$ that $k \in \mathcal{K}_{g}$. Due to the orthogonality of pilots we have $[\mathbf{\Psi}_{un},\mathbf{\Psi}_{mu}]^{H} [\mathbf{\Psi}_{un},\mathbf{\Psi}_{mu}] = \tau \mathbf{I}_{U+G}$, which enforces $ (U+G) \leq \tau \leq T$. Based on these conventions, the received uplink training signal at the BS is
\begin{align}
	\mathbf{Y} = \sum_{u=1}^{U} \sqrt{\tau p_{u}^{up}} \; \mathbf{f}_{u} \boldsymbol{\psi}_{un-u}^{T} + \sum_{g=1}^{G} \sum_{k=1}^{K_{g}} \sqrt{\tau q_{gk}^{up}} \;  \mathbf{g}_{gk} \boldsymbol{\psi}_{mu-g}^{T} + \mathbf{N}
\end{align}
where $\mathbf{N} \in \mathbb{C}^{N \times \tau}$ is the normalized additive noise with $[\mathbf{N}]_{ts} \sim \mathcal{CN}(0,1)$, $p_{u}^{up}$ is proportional to the uplink pilot power that has been used by unicast UT $u$, and $q_{jk}^{up}$ is proportional to the uplink pilot power that has been used by multicast UT $k$ of multicasting group $j$. By correlating $\mathbf{Y} $ with $\boldsymbol{\psi}_{un-u}^{*}$, we obtain the following
\begin{align}
\mathbf{y}_u = \sqrt{\tau p_u^{up}} \mathbf{f}_{u} + \mathbf{n}_u
\end{align}
where $\mathbf{n}_u \sim \mathcal{CN}(\mathbf{0},\mathbf{I}_N)$ is the normalized additive noise. Hence, the BS can estimate $\mathbf{f}_u$ using MMSE estimation as follows \cite[Sec. 3.1.3]{marzetta2016fundamentals}
\begin{align}
	\hat{\mathbf{f}}_{u} = \dfrac{\sqrt{\tau p_{u}^{up}} \beta_{u}}{1 + \tau p_{u}^{up} \beta_{u}} \left( \sqrt{\tau p_{u}^{up}} \; \mathbf{f}_{u} + \mathbf{n}_{u} \right) 
\end{align}
with $\mathbf{n}_{u} \sim \mathcal{CN}(\mathbf{0},\mathbf{I}_{N})$ is the normalized additive noise. Therefore, $\hat{\mathbf{f}}_{u} \sim \mathcal{CN}(\mathbf{0},\vartheta_{u} \mathbf{I}_{N})$ with $\vartheta_{u} = \dfrac{\tau p_{u}^{up} \beta_{u}^{2}}{1 + \tau p_{u}^{up} \beta_{u}}$ and the channel estimation error is $\tilde{\mathbf{f}}_{u} = \hat{\mathbf{f}}_{u} - \mathbf{f}_{u} \sim \mathcal{CN}(\mathbf{0},(\beta_{u} - \vartheta_{u})\mathbf{I}_{N}) $. We stack the estimated channel vectors between the BS and the $U$ unicast UTs in an $N \times U$ matrix $\hat{\mathbf{F}}= [\hat{\mathbf{f}}_{1}, \ldots,\hat{\mathbf{f}}_{U}]$. Using the same approach, i.e., correlating $\mathbf{Y} $ with $\boldsymbol{\psi}_{mu-g}^{*}$ and then using MMSE estimation \cite[Sec. 12]{kay1993fundamentals}, we can estimate $\mathbf{g}_{gk}$ with
\begin{align}
\label{est.mu}
	\hat{\mathbf{g}}_{gk} = \dfrac{\sqrt{\tau q_{gk}^{up}} \; \eta_{gk}}{1 + \sum_{t=1}^{K_{g}} \tau q_{gt}^{up} \eta_{gt}}  \left( \sum_{t=1}^{K_{g}} \sqrt{ \tau q_{gt}^{up} } \; \mathbf{g}_{gt} + \mathbf{n}_{g} \right)
\end{align}
where $\mathbf{n}_{g} \sim \mathcal{CN}(\mathbf{0},\mathbf{I}_{N})$ is the normalized additive noise. Therefore, $\hat{\mathbf{g}}_{gk} \sim \mathcal{CN}(\mathbf{0},\xi_{gk} \mathbf{I}_{N})$ with $\xi_{gk} =  (\tau q_{gk}^{up} \; \eta_{gk}^{2}) / ( 1 + \sum_{t=1}^{K_{g}} \tau q_{gt}^{up} \eta_{gt} )$. Also the channel estimation error is $\tilde{\mathbf{g}}_{gk} = \hat{\mathbf{g}}_{gk} - \mathbf{g}_{gk} \sim \mathcal{CN}(\mathbf{0},(\eta_{gk} - \xi_{gk})\mathbf{I}_{N}) $. Given \eqref{est.mu}, it is obvious that for every multicast group $g$ the channels estimates of the multicast UTs are equivalent up to a scalar coefficient. We also estimate $\mathbf{g}_{g} = \sum_{t=1}^{K_{g}} \sqrt{ \tau q_{gt}^{up} } \; \mathbf{g}_{gt}$, which is a linear combination of the channels of all multicast UTs within group $g$ and we have \cite{kay1993fundamentals,MeysamMMFTWC}
\begin{align}
\label{ghat}
	\hat{\mathbf{g}}_{g} =  \dfrac{ \sum_{t=1}^{K_{g}} \tau q_{gt}^{up} \eta_{gt}}{1 + \sum_{t=1}^{K_{g}} \tau q_{gt}^{up} \eta_{gt}}   \left( \sum_{t=1}^{K_{g}} \sqrt{ \tau q_{gt}^{up} } \; \mathbf{g}_{gt} + \mathbf{n}_{g} \right)
\end{align}
where $\hat{\mathbf{g}}_{g} \sim \mathcal{CN}(\mathbf{0}, \gamma_{g} \mathbf{I}_{N})$ with $\gamma_{g} = \dfrac{ \left( \sum_{t=1}^{K_{g}} \tau q_{gt}^{up} \eta_{gt} \right)^{2} }{1 + \sum_{t=1}^{K_{g}} \tau q_{gt}^{up} \eta_{gt}} $. We stack the $G$ composite channel vectors in an $N \times G$ matrix $\hat{\mathbf{G}} = [\hat{\mathbf{g}}_{1},\ldots,\hat{\mathbf{g}}_{G}]$. Note that $\hat{\mathbf{g}}_{gk}$ and $\hat{\mathbf{g}}_{g}$ are equal up to a scalar coefficient, i.e.,
\begin{align}
\label{ghat.to.ghatjk}
	\hat{\mathbf{g}}_{gk} = \dfrac{\sqrt{\tau q_{gk}^{up}} \; \eta_{gk}}{\sum_{t=1}^{K_{g}} \tau q_{gt}^{up} \eta_{gt}} \hat{\mathbf{g}}_{g}.
\end{align}

\subsection{Downlink Transmission}
Let us denote the data symbols for the $U$ unicast UTs as $\mathbf{x}=[x_{1},\ldots,x_{U}]^{T}$, where $\mathbf{x} \sim \mathcal{CN}(\mathbf{0},\mathbf{I}_{U})$. We denote the data symbols for the $G$ multicasting groups as $\mathbf{s}=[s_{1},\ldots,s_{G}]^{T}$, where $\mathbf{s} \sim \mathcal{CN}(\mathbf{0},\mathbf{I}_{G})$. Moreover, we assume $\mathbf{x}$ and $\mathbf{s}$ are independent. Then the signal received by $m$th unicast UT is
\begin{align}
\label{unicast-RXsignal}
	y_{m} = \mathbf{f}_{m}^{H} \left( \mathbf{V} \mathbf{x} + \mathbf{W} \mathbf{s} \right) + n_{m}
\end{align}
where $n_{m} \sim \mathcal{CN}(0,1)$ is the normalized additive noise, $\mathbf{V} = [\mathbf{v}_{1},\ldots,\mathbf{v}_{U}]$ is the $N \times U$ unicast precoding matrix with $\mathbf{v}_{m}$ being the precoding vector of $m$th unicast UT, and $\mathbf{W}=[\mathbf{w}_{1},\ldots,\mathbf{w}_{G}]$ is the $N \times G$ multicast precoding matrix with $\mathbf{w}_{j}$ being the precoding vector of $j$th multicasting group. The received signal by $k$th multicast UT in $j$th multicasting group is
\begin{align}
\label{multicast-RXsignal}
	z_{jk} = \mathbf{g}_{jk}^{H} \left( \mathbf{W} \mathbf{s} + \mathbf{V} \mathbf{x} \right) + n_{jk}
\end{align}
where $n_{jk} \sim \mathcal{CN}(0,1)$ is the normalized additive noise. Therefore the precoding matrix becomes an $N \times (U+G)$ matrix $[\mathbf{V},\mathbf{W}]$. We will detail the structure of $\mathbf{V}$ and $\mathbf{W}$ next.

\subsection{Precoders Structures}
The BS should perform precoding for both its unicast UTs and multicast UTs of the system. As simple linear precoding schemes like MRT and ZF provide close-to-optimal performance in single-cell massive MIMO systems \cite{marzetta2010noncooperative,hoydis2013massive,LSAPowerNorm}, and due to their good performance for multicast transmission \cite{MeysamMMFTWC}, in the sequel we consider only the MRT and ZF precoding schemes.

\subsubsection{MRT Precoding}
The precoding vector for $m$th unicast UT is
\begin{align}
\label{MRT-unicast}
	\mathbf{v}_{m}^{\mathrm{MRT}}=\sqrt{\frac{p_{m}^{dl}}{N \vartheta_{m}}} \;\; \hat{\mathbf{f}}_{m}  
\end{align}
where $p_{m}^{dl}$ is the downlink power of the unicast precoding vector, i.e., $\mathbb{E}[\Vert \mathbf{v}_{m} \Vert^{2}] = p_{m}^{dl}$. Also the precoding vector for $j$th multicast group is
\begin{align}
\label{MRT-multicast}
	\mathbf{w}_{j}^{\mathrm{MRT}} = \sqrt{\frac{q_{j}^{dl}}{N \gamma_{j}}} \;\; \hat{\mathbf{g}}_{j}
\end{align}
where $q_{j}^{dl}$ is the downlink power of the multicast precoding vector, i.e., $\mathbb{E}[\Vert \mathbf{w}_{j} \Vert^{2}] = q_{j}^{dl}$. We assume the total downlink available power at the BS is $P_{tot}$. Then the power allocated to the downlink unicast and multicast precoding vectors should meet the  condition
\begin{align}
\label{TotalPower}
     P_{un} + P_{mu} \leq P_{tot}
\end{align}
where $ P_{un}=\sum_{m=1}^{U} p_{m}^{dl}$ and $ P_{mu} = \sum_{j=1}^{G} q_{j}^{dl}$ are the precoding powers used for unicast and multicast transmissions at the BS, respectively.

\subsubsection{ZF Precoding}
Denote $\hat{\mathbf{C}} = [\hat{\mathbf{F}},\hat{\mathbf{G}}]$. The ZF precoding vectors for $m$ unicast UT and the $K_{j}$ multicast UTs in $j$th multicast group are
\begin{align}
\label{zf.uni}
\mathbf{v}_{m}^{\mathrm{ZF}} &= \sqrt{(N-G-U) p_{m}^{dl} \vartheta_{m}} \; \hat{\mathbf{C}} \left( \hat{\mathbf{C}}^{H} \hat{\mathbf{C}} \right)^{-1} \mathbf{e}_{m,U+G}
\\
\label{zf.mul}
\mathbf{w}_{j}^{\mathrm{ZF}} &=  \sqrt{(N-G-U) q_{j}^{dl} \gamma_{j}} \; \hat{\mathbf{C}} \left( \hat{\mathbf{C}}^{H} \hat{\mathbf{C}} \right)^{-1} \mathbf{e}_{U+j,U+G}
\end{align} 
where $\mathbf{e}_{i,U+G}$ is the $i$th column of $\mathbf{I}_{U+G}$, $\mathbb{E}[\Vert \mathbf{v}_{m}^{\mathrm{ZF}} \Vert^{2}] = p_{m}^{dl}$ and $\mathbb{E}[\Vert \mathbf{w}_{j}^{\mathrm{ZF}} \Vert^{2}] = q_{j}^{dl}$. Note that $\{ p_{m}^{dl} \}$ and $\{ q_{j}^{dl} \}$ should satisfy \eqref{TotalPower}.


\section{Achievable Spectral Efficiencies}
In this section, we use a standard capacity bounding technique from the massive MIMO literature to derive rigorous achievable SEs of the UTs \cite{jose2011pilot,larsson2016joint,hoydis2013massive,ngo2013energy}. 
More precisely, each UT uses the signal received over its average effective channel (e.g., $\mathbb{E}[\mathbf{f}_{m}^{H}  \mathbf{v}_{m}]$ for unicast UT $m$) for signal detection while treating all other terms as noise \cite[Sec. 2.3]{marzetta2016fundamentals}, \cite[Sec. 4.3]{bjornson2017massive}. Note that this bound does not require channel hardening and it is an information theoretic technique that goes back to (at least) \cite{medard}.

\subsection{Achievable Spectral Efficiency for MRT Precoding}
Consider the unicast UTs. Starting from \eqref{unicast-RXsignal}, we can write the signal received by $m$th unicast UT as
\begin{align}
\label{uatf}
y_{m} =& \; \mathbb{E}[\mathbf{f}_{m}^{H}  \mathbf{v}_{m}] x_{m} + (\mathbf{f}_{m}^{H}  \mathbf{v}_{m} - \mathbb{E}[\mathbf{f}_{m}^{H}  \mathbf{v}_{m}])  x_{m}   + \!\! \sum_{u=1, u \neq m}^{U}  \mathbf{f}_{m}^{H}  \mathbf{v}_{u} x_{u} + \sum_{g=1}^{G} \mathbf{f}_{m}^{H} \mathbf{w}_{g} s_{g}  + n_{m}.
\end{align}
Now by applying the bounding technique \cite[Sec. 2.3.4]{marzetta2016fundamentals}, which treats the first term of \eqref{uatf} as the desired signal received over a deterministic channel and the remaining terms as effective noise, the following SE is achievable for the unicast UT $m$, for any precoding scheme
\begin{align}
\label{SE.MRT.un.gen}
\mathrm{SE}_{m,un} = \left(1 - \dfrac{\tau}{T} \right) \log_{2}(1 + \mathrm{SINR}_{m,un})
\end{align}
where
\begin{align}
\label{SINR.MRT.un.gen}
	\mathrm{SINR}_{m,un} =\dfrac{\vert \mathbb{E}[\mathbf{f}_{m}^{H}  \mathbf{v}_{m}] \vert^{2}}{1 - \vert \mathbb{E}[\mathbf{f}_{m}^{H}  \mathbf{v}_{m}] \vert^{2} + \sum_{u=1}^{U} \mathbb{E} [\vert \mathbf{f}_{m}^{H}  \mathbf{v}_{u} \vert^{2}] + \sum_{g=1}^{G} \mathbb{E} [\vert \mathbf{f}_{m}^{H} \mathbf{w}_{g}\vert^{2}] }. 
\end{align}
Insert \eqref{MRT-unicast} and \eqref{MRT-multicast} into \eqref{SINR.MRT.un.gen}. For the nominator we have
\begin{align}
	\mathbb{E}[\mathbf{f}_{m}^{H}  \mathbf{v}_{m}^{\mathrm{MRT}}] = \! \sqrt{\frac{p_{m}^{dl}}{N \vartheta_{m}}} \mathbb{E}[\mathbf{f}_{m}^{H} \hat{\mathbf{f}}_{m} ] = \! \sqrt{N \vartheta_{m} p_{m}^{dl}}.
\end{align}
Let us denote $\phi_{u} = p_{u}^{dl} / \left(N (1 +  \tau p_{u}^{up} \beta_{u})\right) $. Now considering the terms $\mathbb{E} [\vert \mathbf{f}_{m}^{H}  \mathbf{v}_{u}^{\mathrm{MRT}} \vert^{2}]$ we have
\begin{align}
\mathbb{E} [\vert \mathbf{f}_{m}^{H}  \mathbf{v}_{u}^{\mathrm{MRT}} \vert^{2}] =  \phi_{u} \; \mathbb{E} \left[ \mathbf{f}_{m}^{H}   \left( \! \sqrt{\tau p_{u}^{up}} \mathbf{f}_{u} \! + \! \mathbf{n}_{u} \right) \!\! \left( \! \sqrt{\tau p_{u}^{up}} \mathbf{f}_{u} \! + \! \mathbf{n}_{u} \right)^{\!H}  \!\! \mathbf{f}_{m} \right]\!. 
\end{align}
If $u \neq m$ then we have $\mathbb{E} [\vert \mathbf{f}_{m}^{H}  \mathbf{v}_{u} \vert^{2}] = p_{u}^{dl} \beta_{m} $. Otherwise ($u=m$), we have
\begin{align}
\label{inproof11}
\mathbb{E} [\vert \mathbf{f}_{m}^{H}  \mathbf{v}_{m}^{\mathrm{MRT}} \vert^{2}] \stackrel{(a)}{=} \phi_{m} \mathrm{tr}  \left( \mathbb{E} \left[ \mathbf{f}_{m} \mathbf{f}_{m}^{H} \mathbf{n}_{u}  \mathbf{n}_{u}^{H} \right] +  \tau p_{m}^{up} \mathbb{E} \left[ ( \mathbf{f}_{m} \mathbf{f}_{m}^{H})^{2} \right] \right) 
\stackrel{(b)}{=} p_{m}^{dl} \beta_{m} + N p_{m}^{dl} \vartheta_{m}
\end{align}
where $(a)$ is achieved by omitting independent terms that are resulting in zero expected value and $(b)$ follows from \cite[Lemma 2.10]{tulino2004random}. Now let us calculate the terms $\mathbb{E} [\vert \mathbf{f}_{m}^{H} \mathbf{w}_{g}^{\mathrm{MRT}} \vert^{2}]$. We have $\mathbb{E} [\vert \mathbf{f}_{m}^{H} \mathbf{w}_{g}^{\mathrm{MRT}} \vert^{2}] = \dfrac{q_{g}^{dl}}{N \gamma_{g}} \mathbb{E}[\mathbf{f}_{m}^{H} \hat{\mathbf{g}}_{g} \hat{\mathbf{g}}_{g}^{H}\mathbf{f}_{m}] = q_{g}^{dl} \beta_{m}$. So
\begin{align}
	\sum_{g=1}^{G} \mathbb{E} [\vert \mathbf{f}_{m}^{H} \mathbf{w}_{g}^{\mathrm{MRT}}\vert^{2}] = \sum_{g=1}^{G} q_{g}^{dl}  \beta_{m} = \beta_{m} P_{mu}.
\end{align}
Inserting the obtained results into \eqref{SE.MRT.un.gen} and \eqref{SINR.MRT.un.gen}, the achievable SE and the effective SINR of the unicast UT $m$ with MRT precoding are
\begin{align}
\label{SE.MRT.un}
\mathrm{SE}_{m,un}^{\mathrm{MRT}} &= \left(1 - \dfrac{\tau}{T} \right) \log_{2}(1+\mathrm{SINR}_{m,un}^{\mathrm{MRT}})
\\
\label{SINR.MRT.un}
\mathrm{SINR}_{m,un}^{\mathrm{MRT}} &= \dfrac{N p_{m}^{dl} \vartheta_{m} }{1  + \beta_{m} ( P_{un} + P_{mu})}.
\end{align}
Note that \eqref{SE.MRT.un} and \eqref{SINR.MRT.un}, explicitly explain the relation between the achievable SE of each unicast UT and the system parameters, while accounting for channel estimation and the imperfections of the multiplexing. Also, they have similar structure as the standard results for unicast transmission \cite[Chapter 3]{marzetta2016fundamentals}.

Now let us consider the multicast UTs. Starting from \eqref{multicast-RXsignal}, the signal received by $k$th UT in $j$th multicasting group can be written as follows
\begin{align}
z_{jk} =& \; \mathbb{E}[\mathbf{g}_{jk}^{H} \mathbf{w}_{j}] s_{j}
+ \left( \mathbf{g}_{jk}^{H} \mathbf{w}_{j} - \mathbb{E}[\mathbf{g}_{jk}^{H} \mathbf{w}_{j}] \right) s_{j} + \sum_{g=1, g \neq j}^{G} \mathbf{g}_{jk}^{H} \mathbf{w}_{g} s_{g} + \sum_{u=1}^{U} \mathbf{g}_{jk}^{H} \mathbf{v}_{u} x_{u}  + n_{jk}.
\end{align}
By applying the same bounding technique as \cite[Sec. 2.3]{marzetta2016fundamentals}, the following SE is achievable for the multicast UT $k$ in $j$th multicasting group
\begin{align}
\label{SE.MRT.mu.gen}
\mathrm{SE}_{jk,mu} = \left(1 - \dfrac{\tau}{T} \right) \log_{2}(1 + \mathrm{SINR}_{jk,mu})
\end{align}
where
\begin{align}
	\label{SINR.MRT.mu.gen}
	\mathrm{SINR}_{jk,mu} = \dfrac{ \left| \mathbb{E}[\mathbf{g}_{jk}^{H} \mathbf{w}_{j} ] \right| ^{2}  }{ 1 - \left| \mathbb{E}[\mathbf{g}_{jk}^{H} \mathbf{w}_{j}]  \right| ^{2} + \sum_{g=1}^{G} \mathbb{E} \left[ \left| \mathbf{g}_{jk}^{H} \mathbf{w}_{g} \right|^{2} \right] +  \sum_{u=1}^{U} \mathbb{E} [ \vert \mathbf{g}_{jk}^{H} \mathbf{v}_{u} \vert^{2} ] }. 
\end{align}
Insert \eqref{MRT-unicast} and \eqref{MRT-multicast} into \eqref{SINR.MRT.mu.gen}. Define  $\lambda_{g} = \dfrac{q_{g}^{dl} }{1+\sum_{t=1}^{K_{g}} \tau q_{gt}^{up} \eta_{gt}}$. Then for the terms $\mathbb{E}[\mathbf{g}_{jk}^{H} \mathbf{w}_{j}^{\mathrm{MRT}}]$, we have
\begin{align}
\label{ExpectedVal}
\mathbb{E}[\mathbf{g}_{jk}^{H} \mathbf{w}_{j}^{\mathrm{MRT}}] &=\sqrt{\frac{\lambda_{j}}{N}} \; \mathbb{E} \left[ \mathbf{g}_{jk}^{H} \left( \sum_{t=1}^{K_{j}}  \sqrt{\tau q_{jt}^{up}}  \mathbf{g}_{jt} + \mathbf{n}_{j} \right) \right] = \sqrt{N \tau  \lambda_{j} q_{jk}^{up}} \eta_{jk} = \sqrt{N q_{j}^{dl} \xi_{jk} }.
\end{align}
Based on \eqref{ExpectedVal}, the desired signal power of a typical UT $k$ in $j$th multicast group becomes $N q_{j}^{dl} \xi_{jk}$. Note that $\xi_{jk}$ quantitatively express the effect of using a shared pilot for all the multicast UTs in group $j$. More precisely, it clearly determines how the desired signal power of UT $k$ in group $j$ reduces by adding more UTs in this multicasting group. Also for the terms $ \sum_{g=1}^{G} \mathbb{E} \left[ \left| \mathbf{g}_{jk}^{H} \mathbf{w}_{g}^{\mathrm{MRT}} \right|^{2} \right]$ we have
\begin{align}
\label{aidmeysam1}
\sum_{g=1}^{G} \mathbb{E} \left[ \left| \mathbf{g}_{jk}^{H} \mathbf{w}_{g}^{\mathrm{MRT}} \right|^{2} \right]
\stackrel{(c)}{=} \sum_{g=1}^{G}  \frac{\lambda_{g}}{N}  \left( \sum_{t=1}^{K_{g}} \sum_{t^{\prime}=1}^{K_{g}} \tau \sqrt{ q_{gt}^{up} q_{gt^{\prime}}^{up}} \mathbb{E} \left[ \mathbf{g}_{jk}^{H} \mathbf{g}_{gt} \mathbf{g}_{gt^{\prime}}^{H} \mathbf{g}_{jk} \right] +  \mathbb{E} [ \mathbf{g}_{jk}^{H}  \mathbf{n}_{g} \mathbf{n}_{g}^{H} \mathbf{g}_{jk} ]\right)  
\end{align}
where $(c)$ is obtained by utilizing independency of noise and UTs' channels, and inserting \eqref{ghat} and \eqref{MRT-multicast}. Continuing from \eqref{aidmeysam1} and after straightforward calculation it reduces to $\eta_{jk} P_{mu} + N \tau  \lambda_{j} q_{jk}^{up}  \eta_{jk}^{2}$. Now let us consider the terms $ \mathbb{E} [ \vert \mathbf{g}_{jk}^{H} \mathbf{v}_{u}^{\mathrm{MRT}} \vert^{2} ]$, we have $\mathbb{E} [ \vert \mathbf{g}_{jk}^{H} \mathbf{v}_{u}^{\mathrm{MRT}} \vert^{2} ] = \frac{p_{u}^{dl}}{N \vartheta_{u}} \mathbb{E} [ \mathbf{g}_{jk}^{H} \hat{\mathbf{f}}_{u}  \hat{\mathbf{f}}_{u}^{H} \mathbf{g}_{jk} ]=  \eta_{jk} p_{u}^{dl}$. Hence $ \sum_{tu=1}^{U} \mathbb{E} [ \vert \mathbf{g}_{jk}^{H} \mathbf{v}_{u}^{\mathrm{MRT}} \vert^{2} ] =  \sum_{u=1}^{U} \eta_{jk} p_{u}^{dl}  = \eta_{jk} P_{un} $. Inserting the obtained results into \eqref{SE.MRT.mu.gen} and \eqref{SINR.MRT.mu.gen}, the achievable SE and the effective SINR of multicast UT $k$ in group $j$ with MRT precoding are
\begin{align}
\label{SE.MRT.mu}
\mathrm{SE}_{jk,mu}^{\mathrm{MRT}} &= \left(1 - \dfrac{\tau}{T} \right) \log_{2}(1+\mathrm{SINR}_{jk,mu}^{\mathrm{MRT}})
\\
\label{SINR.MRT.mu}
\mathrm{SINR}_{jk,mu}^{\mathrm{MRT}} &= \dfrac{N  q_{j}^{dl}  \xi_{jk}  }{1  + \eta_{jk} (P_{mu} + P_{un})}.
\end{align}
Note that \eqref{SE.MRT.mu} (and \eqref{SINR.MRT.mu}) follows similar structure as \eqref{SE.MRT.un} (and \eqref{SINR.MRT.un}). Considering the numerator of \eqref{SINR.MRT.mu} (or \eqref{SINR.MRT.un}), we see that  the desired signal is proportional to $N$, which is due to coherent beamforming. Also, it is proportional to $q_{j}^{dl}  \xi_{jk}$ ($p_{m}^{dl} \vartheta_{m}$), which shows the effect of power allocation and the effective channel gain experienced by the UT. Considering the denominator, we note that the interference experienced by the UT due to joint unicast and multicast transmission has seamlessly concentrated to $\eta_{jk}(P_{mu}+P_{un})$ (or to $(\beta_{m}(P_{mu}+P_{un}))$) and the effect of normalized noise is shown by $1$. The achieved neat and tidy expressions enable us do further analysis and provide much more insights, as will be detailed in Section IV and V.

\subsection{Achievable Spectral Efficiency for ZF Precoder}
Consider the unicast UTs. Starting from \eqref{unicast-RXsignal} and applying \eqref{zf.uni}, \eqref{zf.mul}, and $ \mathbf{f}_{u} = \hat{\mathbf{f}}_{u} - \tilde{\mathbf{f}}_{u}$, we can write the signal received by $m$th unicast UT as
\begin{align}
y_{m}  = \; \sqrt{(N-G-U) p_{m}^{dl}\vartheta_{m}} \; x_{m} - \sum_{u=1}^{U} \tilde{\mathbf{f}}_{m}^{H} \mathbf{v}_{u}^{\mathrm{ZF}} x_{u}  - \sum_{g=1}^{G} \tilde{\mathbf{f}}_{m}^{H} \mathbf{w}_{g}^{\mathrm{ZF}} s_{g} + n_{m}. 
\end{align}
Therefore, the effective SINR of unicast UT $m$ becomes \cite[Sec. 2.3]{marzetta2016fundamentals}
\begin{align}
\mathrm{SINR}_{m,un}^{\mathrm{ZF}} = \dfrac{(N-G-U) p_{m}^{dl}\vartheta_{m}}{1 + \sum_{u=1}^{U} \mathbb{E}[\vert \tilde{\mathbf{f}}_{m}^{H} \mathbf{v}_{u}^{\mathrm{ZF}} \vert^{2}] + \sum_{g=1}^{G}  \mathbb{E}[\vert \tilde{\mathbf{f}}_{m}^{H} \mathbf{w}_{g}^{\mathrm{ZF}} \vert^{2} ] }.
\end{align}
Now for the term $\mathbb{E}[\vert \tilde{\mathbf{f}}_{m}^{H} \mathbf{v}_{u}^{\mathrm{ZF}} \vert^{2}]$ we have
\begin{align}
& \mathbb{E}[\vert \tilde{\mathbf{f}}_{m}^{H} \mathbf{v}_{u}^{\mathrm{ZF}} \vert^{2}] \! \stackrel{(d)}{=} \!
\tr \left( \mathbb{E} [\tilde{\mathbf{f}}_{m} \tilde{\mathbf{f}}_{m}^{H}]  \mathbb{E} [\mathbf{v}_{u}^{\mathrm{ZF}} \mathbf{v}_{u}^{\mathrm{ZF}H}] \right) \! = \! (\beta_{m} - \vartheta_{m}) p_{u}^{dl}
\end{align}

where in $(d)$ we used the independency of $\tilde{\mathbf{f}}_{m}$ and $\mathbf{v}_{u}^{\mathrm{ZF}}$. Now for $\mathbb{E}[\vert \tilde{\mathbf{f}}_{m}^{H} \mathbf{w}_{g} \vert^{2}]$ we have
\begin{align}
& \mathbb{E}[\vert \tilde{\mathbf{f}}_{m}^{H} \mathbf{w}_{g}^{\mathrm{ZF}} \vert^{2}] \! \stackrel{(e)}{=} \! \tr \left( \mathbb{E}[ \tilde{\mathbf{f}}_{m} \tilde{\mathbf{f}}_{m}^{H}] \mathbb{E} [\mathbf{w}_{g}^{\mathrm{ZF}} \mathbf{w}_{g}^{\mathrm{ZFH}}] \right) = (\beta_{m} - \vartheta_{m}) q_{g}^{dl}
\end{align}
where in $(e)$ we used the independency of $\tilde{\mathbf{f}}_{m}$ and $\mathbf{w}_{g}^{\mathrm{ZF}}$. Therefore, the achievable SE and the effective SINR of unicast UT $m$ with ZF precoding are
\begin{align}
\label{SE.ZF.un}
\mathrm{SE}_{m,un}^{\mathrm{ZF}} &= \left(1 - \dfrac{\tau}{T} \right) \log_{2}(1+\mathrm{SINR}_{m,un}^{\mathrm{ZF}})
\\
\label{SINR.ZF.un}
\mathrm{SINR}_{m,un}^{\mathrm{ZF}} &= \dfrac{(N-G-U) p_{m}^{dl}\vartheta_{m}}{1 + (\beta_m - \vartheta_{m}) \left( P_{un} +  P_{mu} \right) }.  
\end{align}
Note that \eqref{SE.ZF.un} and \eqref{SINR.ZF.un} clearly show the relation between the achievable SE and the system parameters. For example, \eqref{SINR.ZF.un} shows that the coherent beamforming gain is $N-G-U$, the effect of power allocation and the effective channel gain experienced by the UT is $p_{m}^{dl}\vartheta_{m}$, and the interference experienced by the UT due to joint unicast and multicast transmission is $(\beta_{m} - \vartheta_{m}) (P_{un} + P_{mu})$.

Now let us consider the multicast UTs. The signal received by the $k$th multicast UT in $j$th multicast group is
\begin{align}
y_{jk} &= \left( \hat{\mathbf{g}}_{jk} - \tilde{\mathbf{g}}_{jk} \right)^{H} \left( \mathbf{V}^{\mathrm{ZF}} \mathbf{x} + \mathbf{W}^{\mathrm{ZF}} \mathbf{s} \right) + n_{jk}  \stackrel{(f)}{=} \varrho_{jk} \; s_{j} - \sum_{u=1}^{U} \tilde{\mathbf{g}}_{jk}^{H} \mathbf{v}_{u}^{\mathrm{ZF}} x_{u} - \sum_{g=1}^{G} \tilde{\mathbf{g}}_{jk}^{H} \mathbf{w}_{g}^{\mathrm{ZF}} s_{g} + n_{jk} 
\end{align}
where in $(f)$ we used \eqref{ghat.to.ghatjk} and $\varrho_{jk} = \eta_{jk} \sqrt{\dfrac{ (N\!-\!G\!-\!U) \tau q_{jk}^{up} q_{j}^{dl} }{1+\sum_{t=1}^{K_{j}} \tau q_{jt}^{up} \eta_{jt}}}$. Therefore under ZF precoding the effective SINR of multicast UT $k$ in group $j$ becomes \cite[Sec. 2.3]{marzetta2016fundamentals} 
\begin{align}
\mathrm{SINR}_{jk,mu}^{\mathrm{ZF}} \!=\! \dfrac{\varrho_{jk}^{2}}{1 \!+ \sum_{u=1}^{U} \mathbb{E}[\vert \tilde{\mathbf{g}}_{jk}^{H} \mathbf{v}_{u}^{\mathrm{ZF}} \vert^{2}] + \sum_{g=1}^{G}  \mathbb{E}[\vert \tilde{\mathbf{g}}_{jk}^{H} \mathbf{w}_{g}^{\mathrm{ZF}} \vert^{2} ] }.
\end{align}
Now for the term $\mathbb{E}[\vert \tilde{\mathbf{g}}_{jk}^{H} \mathbf{v}_{u}^{\mathrm{ZF}} \vert^{2}]$ we have
\begin{align}
\mathbb{E}[\vert \tilde{\mathbf{g}}_{jk}^{H} \mathbf{v}_{u}^{\mathrm{ZF}} \vert^{2}] 
&\stackrel{(g)}{=}
\tr \left( \mathbb{E} \left[ \tilde{\mathbf{g}}_{jk} \tilde{\mathbf{g}}_{jk}^{H} \right] \mathbb{E} \left[ \tilde{\mathbf{v}}_{u}^{\mathrm{ZF}} \tilde{\mathbf{v}}_{u}^{\mathrm{ZF}H} \right] \right) = (\eta_{jk} - \xi_{jk}) p_{u}^{dl}
\end{align}
where $(g)$ is due to the independency of $\tilde{\mathbf{g}}_{jk}$ and $\tilde{\mathbf{v}}_{u}^{\mathrm{ZF}}$. Also for $\mathbb{E}[\vert \tilde{\mathbf{g}}_{jk}^{H} \mathbf{w}_{g}^{\mathrm{ZF}} \vert^{2} ]$ we have
\begin{align}
\mathbb{E}[\vert \tilde{\mathbf{g}}_{jk}^{H} \mathbf{w}_{g}^{\mathrm{ZF}} \vert^{2}] 
&\stackrel{(h)}{=} 
\tr \left( \mathbb{E} \left[ \tilde{\mathbf{g}}_{jk} \tilde{\mathbf{g}}_{jk}^{H} \right] \mathbb{E} \left[ \tilde{\mathbf{w}}_{g}^{\mathrm{ZF}} \tilde{\mathbf{w}}_{g}^{\mathrm{ZF}H} \right] \right) =  (\eta_{jk} - \xi_{jk}) q_{g}^{dl}
\end{align}
where $(h)$ is due to the independency of $\tilde{\mathbf{g}}_{jk}$ and $\tilde{\mathbf{w}}_{g}^{\mathrm{ZF}}$. Therefore, the achievable SE and the effective SINR of UT $k$ in group $j$ with ZF precoding are
\begin{align}
\label{SE.ZF.mu}
\mathrm{SE}_{jk,mu}^{\mathrm{ZF}} &= \left(1 - \dfrac{\tau}{T} \right) \log_{2}(1+\mathrm{SINR}_{jk,mu}^{\mathrm{ZF}})
\\
\label{SINR.ZF.mu}
\mathrm{SINR}_{jk,mu}^{\mathrm{ZF}} &= \dfrac{(N-G-U) q_{j}^{dl}  \xi_{jk} }{ 1 +  (\eta_{jk} - \xi_{jk}) (P_{un} + P_{mu}) }.
\end{align}
Note that \eqref{SINR.ZF.mu} has a similar structure as \eqref{SINR.ZF.un}, and the same remarks apply here. It is also interesting to note the difference between \eqref{SINR.ZF.mu} (or \eqref{SINR.ZF.un}) and \eqref{SINR.MRT.mu} (or \eqref{SINR.MRT.un}). As we switch from \eqref{SINR.MRT.mu} to \eqref{SINR.ZF.un}, the coherent beamforming gain reduces from $N$ to $N-G-U$. This is due to the fact that ZF uses the available degrees of freedom to cancel the interference, by placing the signals in the null-space, at the cost of reducing the beamforming gain. This partially removes the interference, e.g., $\xi_{jk}(P_{un}+P_{mu})$ in \eqref{SINR.ZF.mu}.


\section{Optimal Resource Allocation}
Having closed form expressions for the achievable SEs of unicast and multicast UTs under MRT and ZF precoding, it is interesting to see how the resources, i.e., the $T$ symbols of the coherence interval, the uplink power of the UTs, and the downlink power at the BS, should be allocated to achieve an optimal performance. In the context of multicasting, the common metric of interest is MMF, where the objective is to allocate the resources to maximize the minimum SE (or SINR) among all the multicast UTs, subject to a total transmit power constraint. Considering our system model, the MMF problem for multicast UTs is stated as
\begin{align}
\label{MMF}
	\mathcal{P}1: \underset{{\{q_{j}^{dl}\},\{q_{jk}^{up}\}, \tau}}{\text{maximize}}  \; \min_{j \in \mathcal{G}, k \in \mathcal{K}_{j}}  \;\;& \mathrm{SE}_{jk,mu}^{\dagger}
	\\
	s.t. \; & \sum_{j=1}^{G} q_{j}^{dl}   \leq P - P_{un}  \tag{\ref{MMF}-i}
	\\
	& 0 \leq q_{j}^{dl}   \tag{\ref{MMF}-ii}
	\\
	& 0 \leq \tau q_{jk}^{up} \leq E_{jk} \tag{\ref{MMF}-iii}
	\\
	& \tau \in \{U\!\!+\!G,\ldots,T \} \tag{\ref{MMF}-iv}
\end{align}
where $\dagger$ is either $\mathrm{MRT}$ or $\mathrm{ZF}$, $P$ is the total downlink power at the BS, and $P_{un}$ is a given fixed quantity. Also $E_{jk}$ is the maximum energy limit of the multicast UT $k$ in group $j$ per pilot transmission. Hereafter we denote an objective value of $\mathcal{P}1$ that is obtained for a set of feasible decision variables of $\mathcal{P}1$ by $O_{mu}$ and we denote its optimal objective value by $O_{mu}^{*}$.

\begin{theorem}
	\label{multicastOpt}
	For a fixed value of $P_{un}$ where $0 \leq P_{un} \leq P$ and MRT precoding, at the optimal solution of $\mathcal{P}1$ all the multicast UTs will have equal SEs, i.e., $O_{mu}^{*}(P_{un}) = \mathrm{SE}_{jk} \; \forall j,k$, which is the optimal objective value of $\mathcal{P}1$ and it is equal to 
	\begin{align}
	\label{O1}
	&O_{mu}^{*} (P_{un}) = \left( 1 - \dfrac{U\!\!+\!G}{T} \right) \log_{2} \left( 1 + \Gamma \right)  
	\end{align}
	where 
		\begin{align}
	\Gamma = \dfrac{N P_{mu}}{ \sum_{j=1}^{G} \frac{1}{\Upsilon_{j}}  + \sum_{j=1}^{G} \sum_{t=1}^{K_{j}} \frac{1}{\eta_{jt}} + P \sum_{j=1}^{G} K_{j}}
	\end{align}
	with $\Upsilon_{j} = \min_{k \in \mathcal{K}_{j}} \frac{E_{jk} \eta_{jk}^{2}}{1+\eta_{jk}P}$ and $P_{mu} = P - P_{un}$. Also the optimal values of decision variables are $\tau^{*} = U\!\!+\!G$, $q_{jk}^{up*} = \frac{1+ \eta_{jk}P}{(U\!\!+\!G)\eta_{jk}^{2}} \Upsilon_{j}$,	
	$q_{j}^{dl*} = \frac{\Gamma}{N \Upsilon_{j}}  (1 + \sum_{t=1}^{K_{j}} x_{jt}^{*} \eta_{jt})$ with $x_{jk}^{*} = \frac{1+ \eta_{jk}P}{\eta_{jk}^{2}} \Upsilon_{j}$.
\end{theorem}
\begin{proof}
	The proof is given in the Appendix A.
\end{proof}

\begin{theorem}
	\label{multicastOptZF}
	For a fixed value of $P_{un}$ where $0 \leq P_{un} \leq P$ and ZF precoding, at the optimal solution of $\mathcal{P}1$ all the multicast UTs will have equal SEs, i.e., $O_{mu}^{*}(P_{un}) = \mathrm{SE}_{jk} \; \forall j,k$, which is the optimal objective value of $\mathcal{P}1$ and it is equal to 
	\begin{align}
	\label{O1ZF}
	O_{mu}^{*} (P_{un}) = \left( 1 - \dfrac{U\!\!+\!G}{T} \right) \log_{2} \! \left( \! 1 \! + \! \dfrac{(N-G-U)P_{mu}}{  \sum_{j=1}^{G} \frac{1}{\Upsilon_{j}}  + \! \sum_{j=1}^{G} \sum_{t=1}^{K_{j}} \frac{1}{\eta_{jt}} + \! P \sum_{j=1}^{G} K_{j} - PG} \right)
	\end{align}
	where $\Upsilon_{j} = \min_{k \in \mathcal{K}_{j}} \frac{E_{jk} \eta_{jk}^{2}}{1+\eta_{jk}P}$ and $P_{mu}=P - P_{un}$. The optimal values of decision variables are $\tau^{*} = U+G$, $q_{jk}^{up*} = \frac{1+ \eta_{jk} P}{(U+G)\eta_{jk}^{2}} \Upsilon_{j}$, and $q_{j}^{dl*} = \left( \sum_{g=1}^{G} \frac{B_g - P}{B_j - P} \right)^{\!-1} \!\! P_{mu}$ where $B_{j} = \frac{1}{\Upsilon_{j}} + \sum_{t=1}^{K_{j}} \frac{1}{\eta_{jt}} + K_{j} P$.
\end{theorem}
\begin{proof}
	The proof follows similar procedure as Theorem \ref{multicastOpt}, and is omitted for brevity.
\end{proof}

Note that for multicast transmission when we switch from MRT (Theorem \ref{multicastOpt}) to ZF (Theorem \ref{multicastOptZF}), the SINR terms change in a particular way. The coherent beamforming gain reduces from $N$ to $N-(G + U)$. Also the interference in the denominator reduces by $PG$. This is due to the fact that ZF uses $G+U$ degrees of freedom to cancel the interference toward other UTs at the cost of reducing its coherent beamforming gain.

\begin{corollary}
	\label{Corollarymu}
	At the optimal solution of $\mathcal{P}1$, for both MRT and ZF precoding schemes, $P_{mu} + P_{un} = P $ and as we increase $P_{mu}$ (by reducing $P_{un}$) the optimal objective value of $\mathcal{P}1$ increases.
\end{corollary}
\begin{proof}
	Based on Theorems \ref{multicastOpt} and \ref{multicastOptZF}, $P_{mu} + P_{nu} = P $. Also from \eqref{O1} and \eqref{O1ZF}, it is obvious that the SINR is linearly increasing with $P_{mu}$. Therefore $O_{mu}^{*}(P_{un})$ is monotonically increasing with $P_{mu}$, which completes the proof.
\end{proof}

For unicast UTs we can optimize the resources considering different metrics. Here we consider the weighted SSE, which is the sum of the weighted SEs of the unicast UTs, as it is a common metric of interest in unicast transmission \cite{SumRateMaxPowerAllTWC,LowComSumRateMassiveMIMO,DualitySumrateGaussianMIMO,SumRateTWC,SumPower,SumRateMaxTSP}. Given our proposed single cell joint unicast and multi-group multicast system, the SSE problem for unicast UTs is 
\begin{align}
\label{sumrate}
\mathcal{P}2: \underset{ \{ p_{m}^{dl} \}, \{ p_{m}^{u} \}, \tau}{\text{maximize}}  & \;\;  \sum_{m=1}^{U} \alpha_{m} \mathrm{SE}_{m,un}^{\dagger} 
\\
s.t. \;\; & \sum_{m=1}^{U} p_{m}^{dl} \leq P - P_{mu} \tag{\ref{sumrate}-i}
\\
& 0 \leq p_{m}^{dl}   \tag{\ref{sumrate}-ii}
\\
& 0 \leq \tau p_{m}^{up} \leq E_{m}      \tag{\ref{sumrate}-iii}
\\
& \tau \in \{U\!\!+\!G,\ldots,T \}        \tag{\ref{sumrate}-iv}
\end{align}
where $\dagger$ is either $\mathrm{MRT}$ or $\mathrm{ZF}$, $P_{mu}$ is a given fixed quantity, $E_m$ is maximum energy limit of the $m$th unicast UT per pilot transmission, and $\alpha_{m}$ is the weight of the SE of this user. Hereafter we denote an objective value of $\mathcal{P}2$ that is obtained for a set of feasible decision variables of $\mathcal{P}2$ by $O_{un}$, and we denote its optimal objective value by $O_{un}^{*}$. Now we have the following Theorem.

\begin{theorem}
	\label{unicastOpt}
	For a fixed value of $P_{mu}$ and $0 \leq P_{mu} \leq P$ and MRT precoding, the optimal solution to problem $\mathcal{P}2$ is $\tau^{*} =U+G$, $p_{m}^{up*} = \frac{E_{m}}{U+G}$, and 
	\begin{align}
	\label{P3Sol}
	p_{m}^{dl*} = \max \left\lbrace  0,\dfrac{\alpha_{m}}{\nu \ln 2 } - \dfrac{1+\beta_{m}P}{N \vartheta_{m}^{*}} \right\rbrace
	\end{align}
	where $\vartheta_{m}^{*}= \frac{E_{m} \beta_{m}^{2}}{1+E_{m} \beta_{m}}$ and $\nu$ is selected to satisfy $P_{un} = P - P_{mu}$. The optimal objective value of $\mathcal{P}2$ becomes
	\begin{align}
	\label{O2}
		O_{un}^{*}(P_{mu}) =  \left(1 - \frac{U+G}{T}\right)  \sum_{m=1}^{U} \alpha_{m} \log_{2} (1 + \dfrac{N p_{m}^{dl*} \vartheta_{m}^{*} }{1  + \beta_{m} P} ).
	\end{align}
\end{theorem}
\begin{proof}
	The proof is given in Appendix B.
\end{proof}

\begin{theorem}
	\label{unicastOptZF}
	For a fixed value of $P_{mu}$ and $0 \leq P_{mu} \leq P$ and ZF precoding, the optimal solution to problem $\mathcal{P}2$ is $\tau^{*} =U+G$, $p_{m}^{up*} = \frac{E_{m}}{U+G}$, and 
	\begin{align}
	\label{P3SolZF}
	p_{m}^{dl*} = \max \left\lbrace  0,\dfrac{\alpha_{m}}{\nu \ln 2 } - \dfrac{1+(\beta_{m}-\vartheta_{m}^{*})P}{(N-G-U) \vartheta_{m}^{*}} \right\rbrace
	\end{align}
	where $\vartheta_{m}^{*}= \frac{E_{m} \beta_{m}^{2}}{1+E_{m} \beta_{m}}$ and $\nu$ is selected to satisfy $P_{un} = P - P_{mu}$. The optimal objective value of $\mathcal{P}2$ becomes
	\begin{align}
	\label{O2ZF}
	O_{un}^{*}(P_{mu}) =  \left(1 - \frac{U+G}{T}\right) \sum_{m=1}^{U} \alpha_{m} \log_{2} \left(1 + \dfrac{(N-G-U) p_{m}^{dl*} \vartheta_{m}^{*} }{1  + (\beta_{m} - \vartheta_{m}^{*}) P} \right). 
	\end{align}
\end{theorem}
\begin{proof}
	Due to similar structure of \eqref{SINR.MRT.un} and \eqref{SINR.ZF.un}, the proof follows similar argument as the proof of Theorem \ref{unicastOpt} and is omitted for brevity.
\end{proof}

Theorems \ref{multicastOpt} to \ref{unicastOptZF} determine how the SE of multicast and unicast UTs are related to the system parameters. For example, they prove that as $N$ increases the SEs of both unicast and multicast UTs improve, which is due to the improved coherent transmission of signals and it is obtained by employing a large-scale antenna array. Moreover, for unicast transmission we have a similar behavior as the one we had for multicast transmission. As we switch from MRT (Theorem \ref{unicastOpt}) to ZF (Theorem \ref{unicastOptZF}), the coherent beamforming gain in the SINR expression of unicast UTs reduces from $N$ to $N-(G + U)$, and the interference in the denominator of the SINR expression reduces by $\vartheta_{m}^{*} P$.

\begin{remark}
	Note that the optimal decision variables and the optimal objective values of Theorems 1 to 4, are based solely on large-scale fading parameters ($\beta_u$ and $\eta_{jk}$), and the basic system parameters, e.g., $G$, $K$, etc. The system parameters are known, and the large-scale fading can be easily estimated \cite{liu2017large}. Moreover, as the large-scale fading parameters are changing much slower than small-scale fading (100-1000 times \cite{Emil10Myth}), once one estimates the scalar large-scale fading coefficients, they are valid for a large number of symbol transmissions and the estimation overhead is negligible.
\end{remark}

\begin{corollary}
	\label{Corollaryun}
	At the optimal solution of $\mathcal{P}2$, for both MRT and ZF precoding schemes, $P_{mu} + P_{un} = P $ and as we increase $P_{un}$ (by reducing $P_{mu}$) the optimal objective value of $\mathcal{P}2$ increases.
\end{corollary}
\begin{proof}
	Based on Theorems \ref{unicastOpt} and \ref{unicastOptZF}, $P_{mu} + P_{un} = P $. Assume for the given $P_{mu}$, the optimal objective is denoted as $O_{un,old}^{*}(P_{mu})$ and is obtained for $\{p_{m}^{dl*}\}$. Now if we increase $P_{un}$ to $P_{un}+\delta$ while reducing $P_{mu}$ to $P_{mu} - \delta$, we can divide this extra power equally among all $U$ unicast UTs, i.e. $p_{m}^{dl*} \to p_{m}^{dl*} + \frac{\delta}{U}$, which gives us an objective value $O_{un,new}$ such that $O_{un,new} > O_{un,old}^{*}(P_{mu})$. Denote the new optimal objective value as $O_{un}^{*}(P_{mu}-\delta)$ (the optimal objective value when the precoding power used for unicast transmission is equal to $P_{un} + \delta$). Then, $O_{un}^{*}(P_{mu}-\delta)$ must be bigger than or equal to $O_{un,new}$. Thus the optimal objective value has increased.
\end{proof}

\section{The Pareto Boundary of Joint Unicast and Multi-group Multicast Massive MIMO Systems}
Based on Corollaries \ref{Corollarymu} and \ref{Corollaryun}, it is obvious that $O_{mu}^{*}(P_{un})$ and $O_{un}^{*}(P_{mu})$ are coupled in a conflicting manner, i.e., $P=P_{un}+P_{mu}$, such that improving $O_{mu}^{*}(P_{un})$ would degrade $O_{un}^{*}(P_{mu})$, and vice versa. Therefore we need to consider a multi-objective optimization framework to understand the jointly achievable operating points. This description enables us to balance the two conflicting objectives and obtain Pareto optimal solutions. To this end, we state a MOOP and derive its Pareto boundary \cite{zadeh1963optimality,bjornson2014multiobjective,marler2004survey}. The MOOP is \cite{marler2004survey}
\begin{align}
\label{MOOP}
\mathcal{M}: \max_{\mathbf{x}} &\quad [O_{mu}(\mathbf{x}), O_{un}(\mathbf{x})]^{T}
\\
s.t. &\quad    \mathbf{x} \in  \mathcal{X}    \tag{\ref{MOOP}-i}
\end{align}
where $\mathbf{x} =  ( \{ q_{j}^{dl} \}, \{ q_{jk}^{up} \},  \{ p_{m}^{dl} \}, \{ p_{m}^{up} \}, \tau )$ and $\mathcal{X}$ is the so called resource bundle \cite{bjornson2014multiobjective}, which is achieved by bundling the feasible space of $\mathcal{P}1$ and $\mathcal{P}2$ as
\begin{align}
\begin{split}
	\mathcal{X} = \bigg\{ \left( \{ q_{j}^{dl} \}, \{ q_{jk}^{up} \},  \{ p_{m}^{dl} \}, \{ p_{m}^{up} \}, \tau \right)  \; | \;\;&  0 \leq q_{j}^{dl}  , 0 \leq \tau q_{jk}^{up} \leq E_{jk},  0 \leq p_{m}^{dl},\; 0 \leq \tau p_{m}^{up} \leq E_{m},...
	\\
	&  P_{un} + P_{mu} \leq P, \tau \in \{ U+G,\ldots,T\}  \bigg\}. 
\end{split}
\end{align}
\begin{remark}
		Note the difference between $O_{mu}(\mathbf{x})$ and $O_{mu}^{*}(P_{un})$. $O_{mu}(\mathbf{x})$ is the objective achieved for $\mathcal{P}1$ given $\mathbf{x} \in \mathcal{X}$, while $O_{mu}^{*}(P_{un})$ is the optimal objective achieved for $\mathcal{P}1$ based on Theorem \ref{multicastOpt} or Theorem \ref{multicastOptZF}, given the unicast power is fixed to $P_{un}$. Similarly, $O_{un}(\mathbf{x})$ is the objective achieved for $\mathcal{P}2$ given $\mathbf{x} \in \mathcal{X}$, while $O_{un}^{*}(P_{mu})$ is the optimal objective achieved for $\mathcal{P}2$ based on Theorem \ref{unicastOpt} or Theorem \ref{unicastOptZF}, given the multicast power is fixed to $P_{mu}$.
\end{remark}

The so-called \textit{attainable objective set} of the MOOP $\mathcal{M}$, given in \eqref{MOOP}, is \cite{bjornson2014multiobjective}
\begin{align}
\mathcal{S} = \{(O_{mu}(\mathbf{x}), O_{un}(\mathbf{x})) | \mathbf{x} \in \mathcal{X} \}.
\end{align}
Now we are ready to define the Pareto boundary and Pareto optimality for $\mathcal{M}$. The strong Pareto boundary of $\mathcal{M}$ is a set $\mathcal{B}_{s}$ containing all the tuples $ (O_{mu}(\mathbf{x}^{*}), O_{un}(\mathbf{x}^{*}))$ such that $\mathbf{x}^{*} \in \mathcal{X}$ and $\nexists \mathbf{y} \in \mathcal{X}$ for which either $ O_{mu}(\mathbf{x}^{*}) < O_{mu}(\mathbf{y})$ and $O_{un}(\mathbf{x}^{*})\leq O_{un}(\mathbf{y})$, or $ O_{mu}(\mathbf{x}^{*}) \leq O_{mu}(\mathbf{y})$ and $O_{un}(\mathbf{x}^{*}) <  O_{un}(\mathbf{y})$. In this case, $\mathbf{x}^{*}$ is called a Pareto optimal point \cite{marler2004survey}. Moreover, a point $\mathbf{z}^{*} \in \mathcal{X}$ is called a weak Pareto optimal point if $\nexists \mathbf{y} \in \mathcal{X}$ such that $O_{mu}(\mathbf{z}^{*}) < O_{mu}(\mathbf{y}) , O_{un}(\mathbf{z}^{*}) < O_{un}(\mathbf{y})$ \cite{marler2004survey}. Note that every strong Pareto optimal point is also a weak Pareto optimal point, but the converse is not true. The set $\mathcal{B}_{w}$ that contains all the tuples $(O_{mu}(\mathbf{z}^{*}), O_{un}(\mathbf{z}^{*}))$ where $\mathbf{z}^{*}$ is a weak Pareto optimal point is called the weak Pareto boundary \cite{bjornson2014multiobjective,marler2004survey}. The strong and weak Pareto boundaries are schematically presented in Fig. \ref{ParetoFig}. We have the following theorem for the Pareto boundary and Pareto optimal points of $\mathcal{M}$.

\begin{figure}[h]
	\centering
	\includegraphics[width=0.6\columnwidth, trim={1.5cm 0cm 0cm 1.5cm},clip]{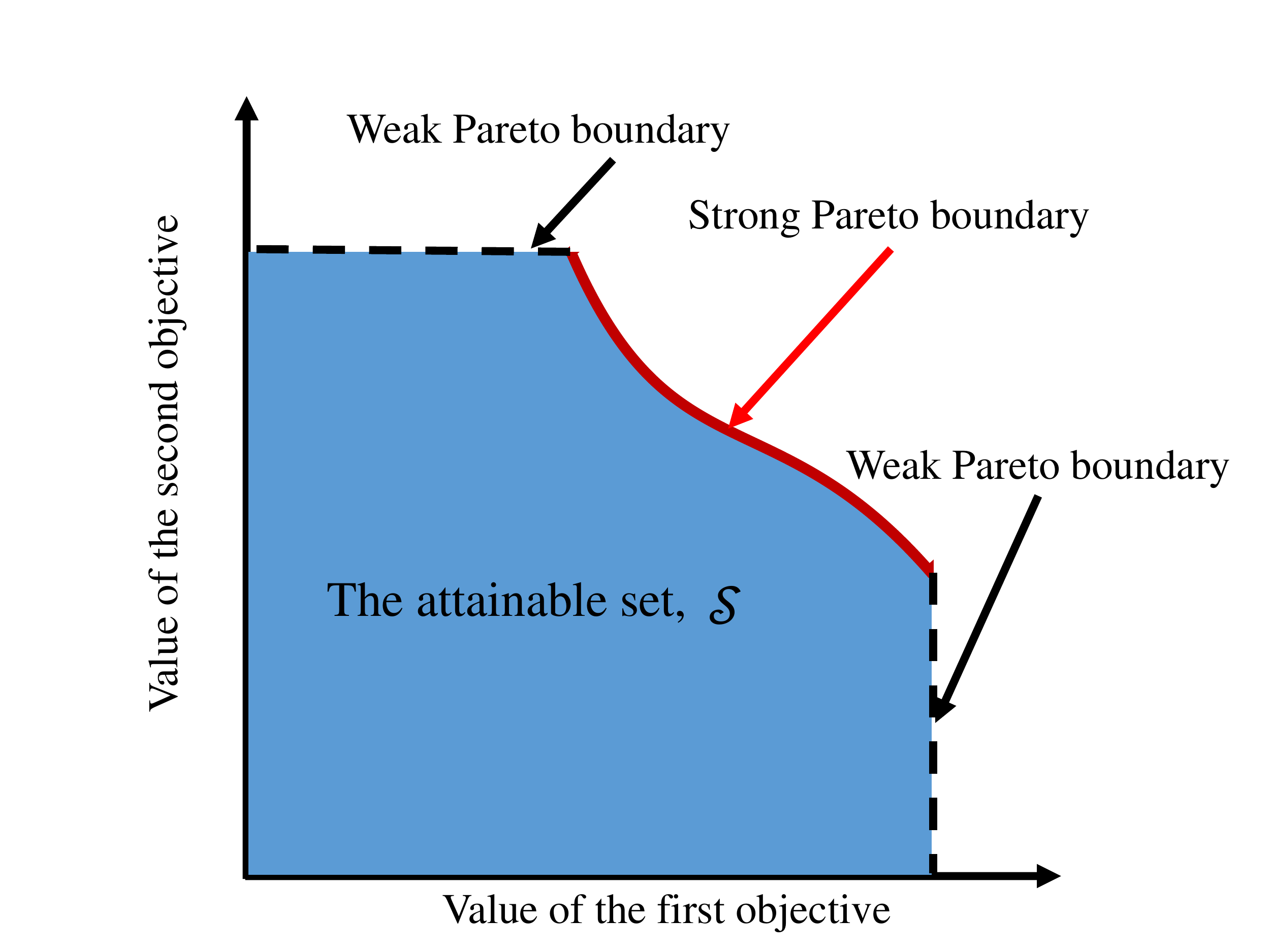}
	\caption{An example of weak and strong Pareto boundaries, and the attainable set \cite{bjornson2014multiobjective}.}
	\label{ParetoFig}
\end{figure}

\begin{theorem}
	\label{Theo_ParetoBound}
	The MOOP \eqref{MOOP} of the considered joint unicast and multi-group multicast massive MIMO system with either MRT or ZF precoding, does not have any weak Pareto optimal points and its strong Pareto boundary is analytically described by
	\begin{align}
		\mathcal{B}_{s} =  \bigg\{ \left(O_{mu}^{*}(P_{un}), O_{un}^{*}(P_{mu})\right) \; | \;  P_{mu} + P_{un} = P, 
		0 \leq P_{mu} \leq P ,  0 \leq P_{un} \leq P \bigg\}.
	\end{align}
	Moreover, $(O_{mu}^{*}(P_{un}), O_{un}^{*}(P_{mu})) \in \mathcal{B}_{s}$ is achieved when $( \{ q_{j}^{dl*} \}, \{ q_{jk}^{up*} \},  \{ p_{m}^{dl*} \}, \{ p_{m}^{up*} \}, \tau^{*} )$ are obtained either from Theorems \ref{multicastOpt} and \ref{unicastOpt} for MRT precoding, or from Theorems \ref{multicastOptZF} and \ref{unicastOptZF} for ZF precoding.
\end{theorem}
\begin{proof}
	The proof is given in Appendix C.
\end{proof}

Theorem \ref{Theo_ParetoBound} is of significant interest due to the following reasons. First, it describes all the Pareto optimal points that can be obtained in the considered system. As each Pareto optimal point describes a particular trade-off between the two objectives (SSE and MMF), it elaborates the set of efficient operating points from which the network designer can select one, which is most desirable for the designer. Second, Theorem \ref{Theo_ParetoBound} not only describes the Pareto boundary points, but also determines the exact values of the system parameters (uplink training powers, downlink transmission powers, and the pilot length.) to achieve such points. Therefore, Theorem \ref{Theo_ParetoBound} exactly describes, in a joint unicast and multi-group multicast massive MIMO system, what can be achieved and how it can be achieved.

\begin{theorem}
	\label{convexity}
	 $\mathcal{S}$, the attainable set of $\mathcal{M}$, is a convex set.
\end{theorem}

\begin{proof}
	We prove the convexity of $\mathcal{S}$, using the definition of a convex set. Consider any two arbitrary tuples $\big( O_{mu}(\mathbf{x}), O_{un}(\mathbf{x})  \big) \in \mathcal{S}$ and $\big( O_{mu}(\mathbf{y}), O_{un}(\mathbf{y})  \big) \in \mathcal{S}$. Define
	\begin{align}
		\big( O_{mu}(\mathbf{\alpha}), O_{un}(\mathbf{\alpha})  \big) &= \alpha \left( O_{mu}(\mathbf{x}), O_{un}(\mathbf{x})  \right) + (1- \alpha) \big( O_{mu}(\mathbf{y}), O_{un}(\mathbf{y})  \big) \quad \alpha \in [0,1].
	\end{align}
	Now we show that $\big( O_{mu}(\mathbf{\alpha}), O_{un}(\mathbf{\alpha})  \big) \in \mathcal{S}$.	Define $P_{mu}(\mathbf{x}) = \sum_{j=1}^{G} q_{j}^{dl} (\mathbf{x})$, $P_{un}(\mathbf{x}) = \sum_{m=1}^{U} p_{m}^{dl} (\mathbf{x})$, $P_{mu}(\mathbf{y}) = \sum_{j=1}^{G} q_{j}^{dl} (\mathbf{y})$, and $P_{un}(\mathbf{y}) = \sum_{m=1}^{U} p_{m}^{dl} (\mathbf{y})$. As $\mathbf{x}, \mathbf{y} \in \mathcal{X}$, $P_{mu}(\mathbf{x}) + P_{un}(\mathbf{x}) \leq P$ and $P_{mu}(\mathbf{y}) + P_{un}(\mathbf{y}) \leq P$. Based on Theorem \ref{multicastOpt} (or Theorem \ref{multicastOptZF}), using $P_{mu}(\alpha) = \alpha P_{mu}(\mathbf{x}) + (1-\alpha)P_{mu}(\mathbf{y})$ for multicast transmission, we have $O^{*}_{mu}\big(P - P_{mu}(\alpha)\big) \geq O_{mu}(\alpha)$. As $P - P_{mu}(\alpha) \geq P_{un}(\alpha)$, based on Theorem \ref{unicastOpt} (or Theorem \ref{unicastOptZF}), using $P - P_{mu}(\alpha)$ for unicast transmission, we have $O^{*}_{un}\big(P_{mu}(\alpha)\big) \geq O_{un}(\alpha)$. Therefore $\forall \alpha \in [0,1]$, $\exists \big(O^{*}_{mu}, O^{*}_{un}\big) \in \mathcal{S} | \big(O^{*}_{mu}, O^{*}_{un}\big) \geq \big(O_{mu}(\alpha), O_{un}(\alpha)\big)$. As \eqref{O1} and \eqref{O2} (or \eqref{O1ZF} and \eqref{O2ZF}) are continuous, by reducing $P_{mu}$ and $\{p_{m}^{dl}\}$ in \eqref{O1} and \eqref{O2} (or \eqref{O1ZF} and \eqref{O2ZF}), we can obtain $\big(O_{mu}(\alpha), O_{un}(\alpha)\big)$, which completes the proof.
\end{proof}

The fact that $\mathcal{S}$ is convex implies that it is optimal to spatially multiplex the unicast and multicast transmission. In contrast, if the set would have been non-convex, time or frequency multiplexing would have been a better solution.


\section{Numerical Results} 
In this section, we use numerical simulations leveraging our proposed results to present more insights into the behavior of joint unicast and multi-group multicast massive MIMO systems. We consider a single cell system where the cell radius is $500$ meters and the unicast and multicast UTs are randomly and uniformly distributed in the cell excluding an inner circle of radius $35$ meters. The large-scale fading coefficient for the multicast UT $k$ in group $j$ is modeled as $\beta_{jk}^{m} = \bar{d}/x_{jk}^{\nu}$, where $\nu=3.76$ is the path-loss exponent, $x_{jk}$ is the distance between the UT and the BS, and $\bar{d} = 10^{-3.5}$ is a constant that regulates the channel attenuation at $35$ meters \cite{3GPPmodel}.\footnote{The considered $\bar{d}$ gives us a pth-loss of $-136.48$ dB at $500$ meters and it can be any other reference distance.} Similarly, the large-scale fading coefficient for the unicast UT $m$ is modeled as $\beta_{m}=\bar{d}/x_{m}^{\nu}$, where $x_{m}$ is the distance between this UT and the BS.

At a carrier frequency of $2$ GHz, the transmission bandwidth is assumed to be $W=20$ MHz, the coherence bandwidth and coherence time are considered to be $200$ kHz and $1$ ms, which results in a coherence interval of length $200$ symbols \cite{marzetta2016fundamentals}. The noise power spectral density is considered to be $\sigma^{2} = -174$ dBm/Hz. The total downlink power $\bar{P}$ is set to $10$ Watts and its corresponding normalized value is $P = \bar{P} / (W \cdot \sigma^2)$. Also, the energy limits $E_m$ and $E_{jk}$ are normalized to $(0.1 \cdot T) / (W \cdot \sigma^2)$. In our simulations we consider a joint unicast and multi-group multicast system with $U=50$ unicast UTs and $G=10$ multicast groups each having $K=100$ multicast UTs, and present the results for both MRT and ZF precoding schemes. For the SSE, the weights are assumed to be equal, e.g., $\alpha_m = 1 \; \forall m \in \mathcal{U}$.

\begin{figure}[]
	\centering
	\begin{subfigure}[b]{0.48\linewidth}
		\centering
		\includegraphics[width=1\columnwidth, trim={1cm 3.5cm 1cm 5cm},clip]{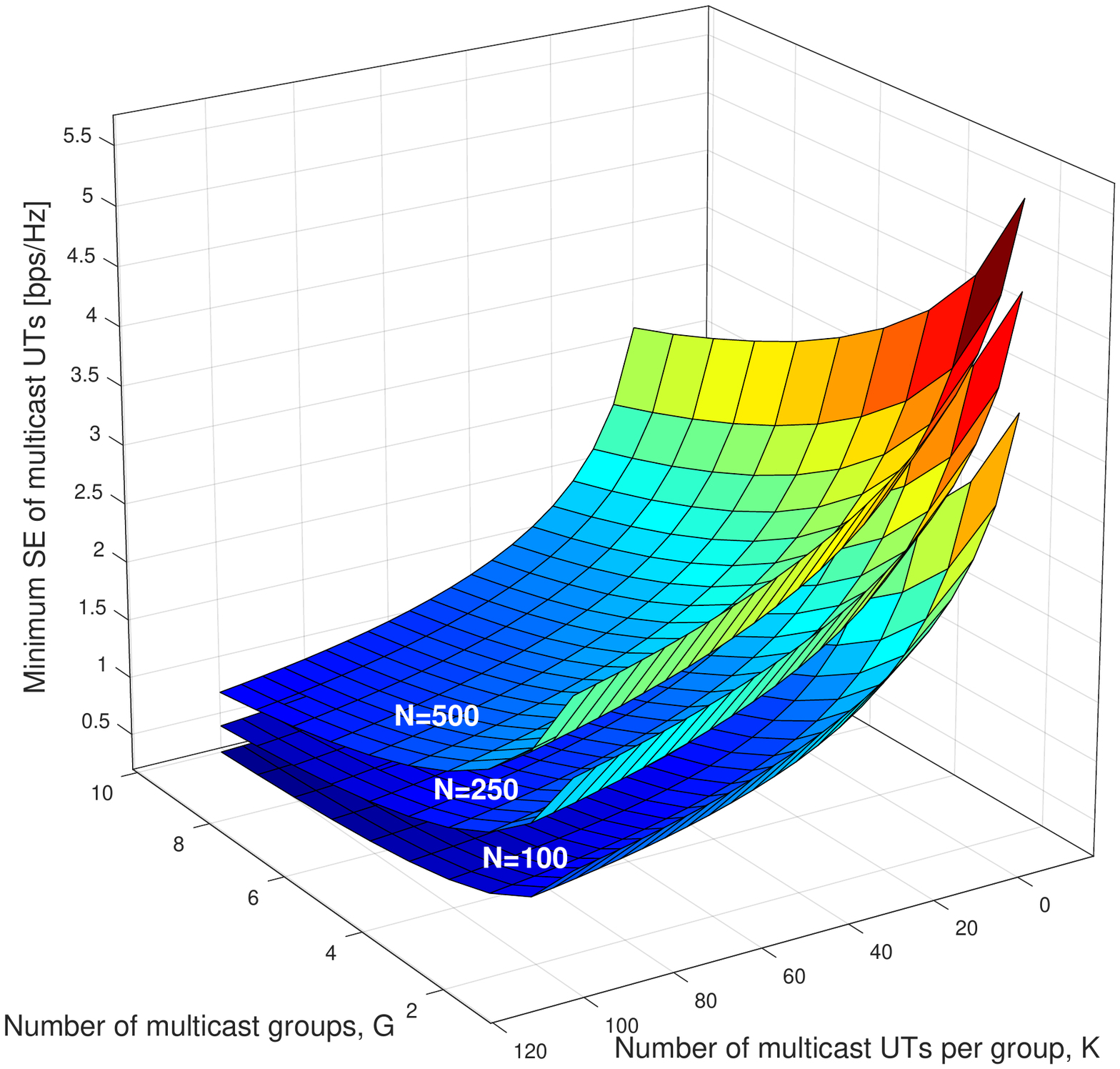}
		\caption{MRT Precoding.}
		\label{FigMuMRT}
	\end{subfigure}
~
	\begin{subfigure}[b]{0.48\linewidth}
		\includegraphics[width=1\columnwidth, trim={1cm 3.5cm 1cm 5cm},clip]{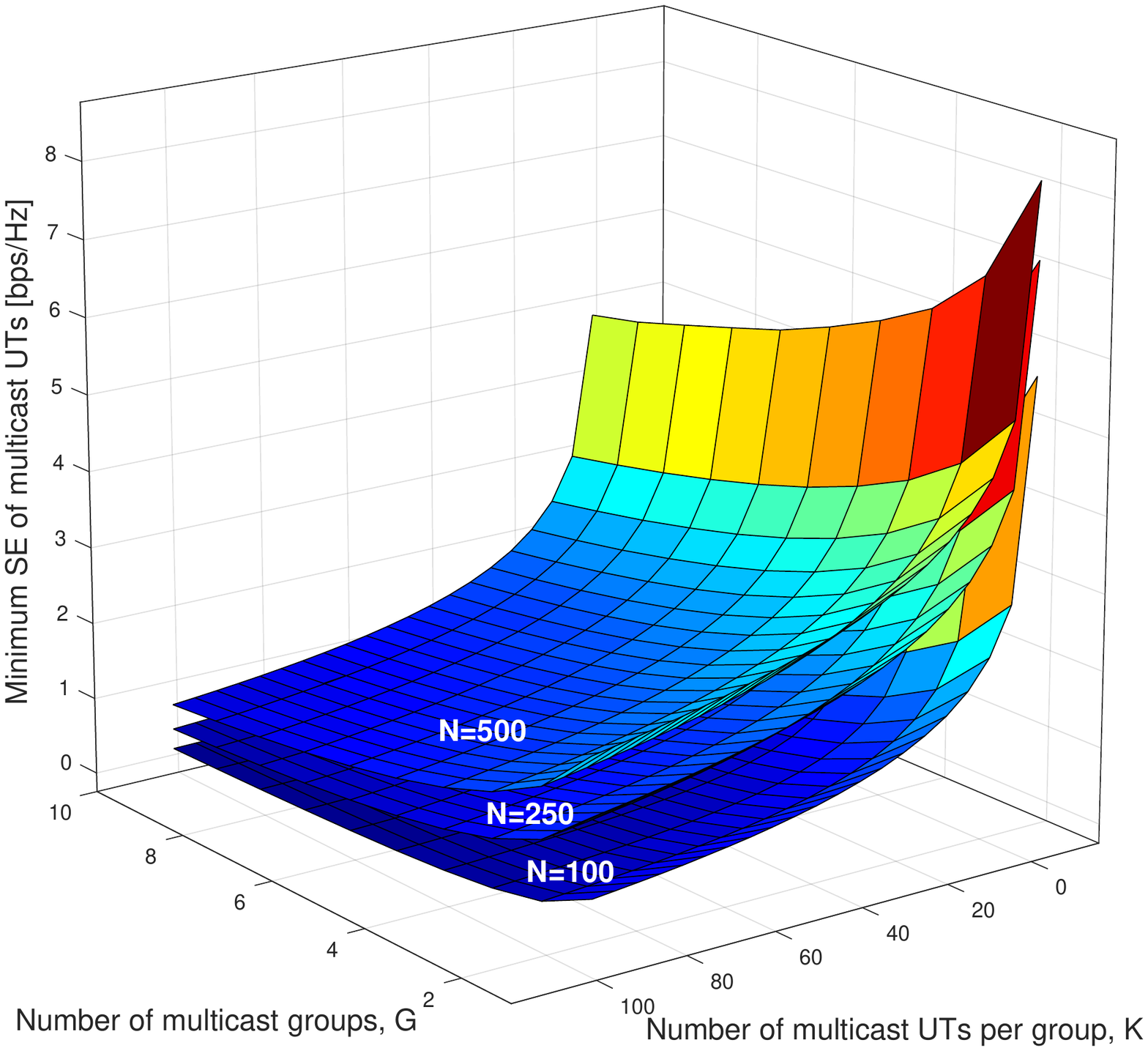}
		\caption{ZF Precoding.}
		\label{FigMuZF}
	\end{subfigure}
\caption{Minimum SE of multicast UTs versus $G$ and $K$.}
\label{FigMu}
\end{figure}

Fig. \ref{FigMu} presents the minimum SE of multicast UTs versus $G$ and $K$, while $P_{un} = P_{mu}$. Three surfaces are plotted, which are associated with $N=100$, $N=250$, and $N=500$. MRT and ZF have similar behavior, for any pair of $(G,K)$ as we increase $N$, we achieve a higher minimum SE for the multicast UTs. This is due to the improved coherent transmission of multicast signals, which is obtained by employing a large-scale antenna array. Also note that the minimum SE of the multicast UTs reduces by increasing $K$ or $G$. However this can be compensated by employing more antennas at the BS, e.g., for MRT precoding the minimum SE of a system with $(N=100,G=4,K=16)$ is equal to a system with $(N=500,G=8,K=46)$.

Fig. \ref{FigUn} presents the SSE of unicast UTs versus $N$ and $U$. For a fixed number of unicast UTs ($U$), the SSE improves as we increase the number of antennas ($N$). However, for a fixed number of antennas, as we increase the number of UTs, the SSE increases and then reduces. Thanks to the closed form expression for SSE of unicast UTs, obtained in Theorem \ref{unicastOpt}, one can easily determine the optimal number of unicast UTs that maximizes the SSE of unicast UTs of the system. Note that for ZF, we need $N> G+U$ to achieve a nonzero SSE (Fig. \ref{FigUnZF}), while the MRT precoder does not have this limitation (Fig. \ref{FigUnMRT}). On the other hand, when we have enough resources, e.g., $N \gg U+G$ and $T \gg U+G $, the SSE of ZF is considerably higher than MRT.

\begin{figure}[]
	\centering
	\begin{subfigure}[b]{0.48\linewidth}
		\centering
		\includegraphics[width=1 \columnwidth, trim={0cm 3cm 0.5cm 4cm},clip]{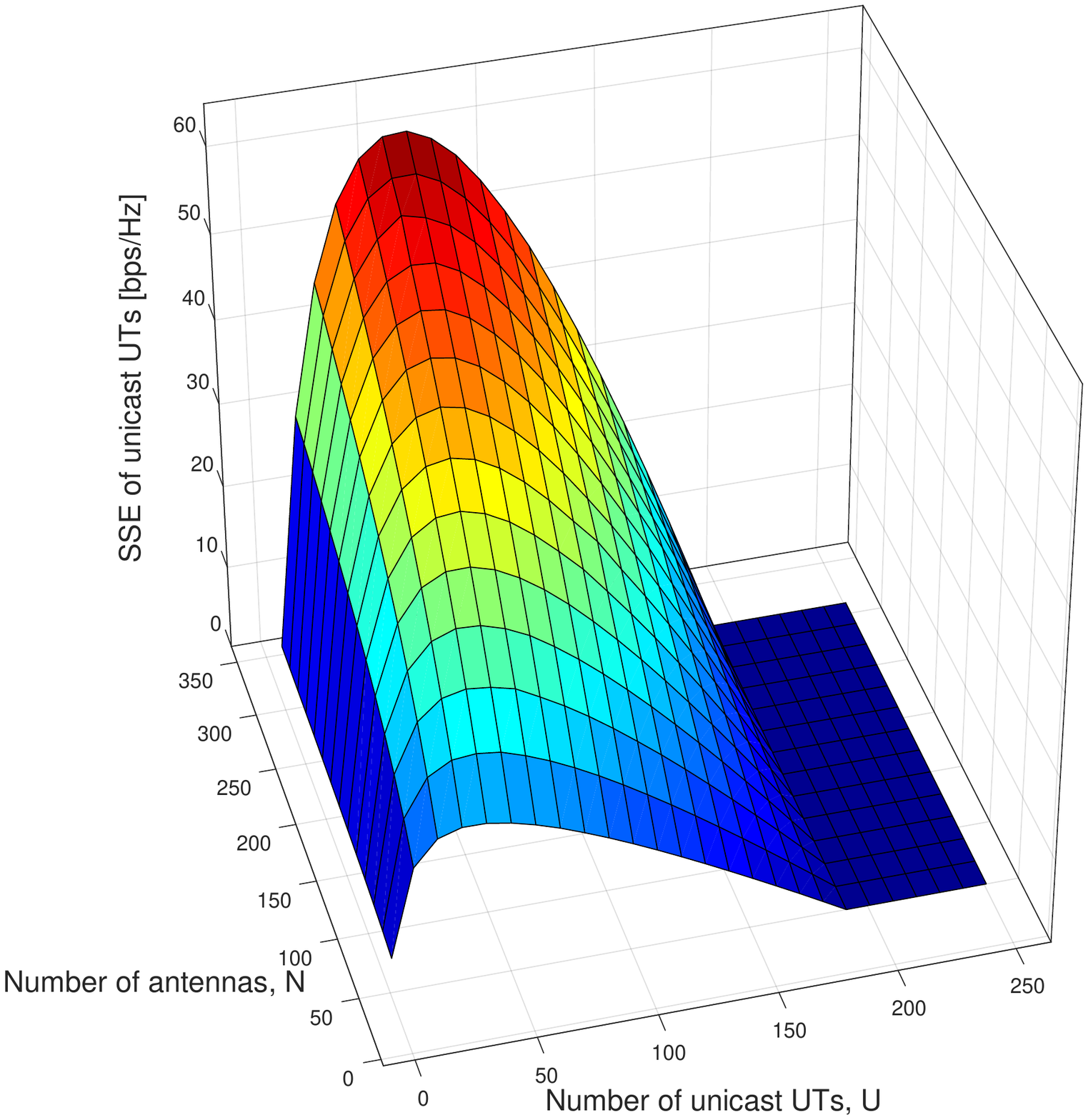}
		\caption{MRT Precoding.}
		\label{FigUnMRT}
	\end{subfigure}
	~
	\begin{subfigure}[b]{0.48\linewidth}
		\centering
		\includegraphics[width=1 \columnwidth, trim={0cm 3cm 0.5cm 4cm},clip]{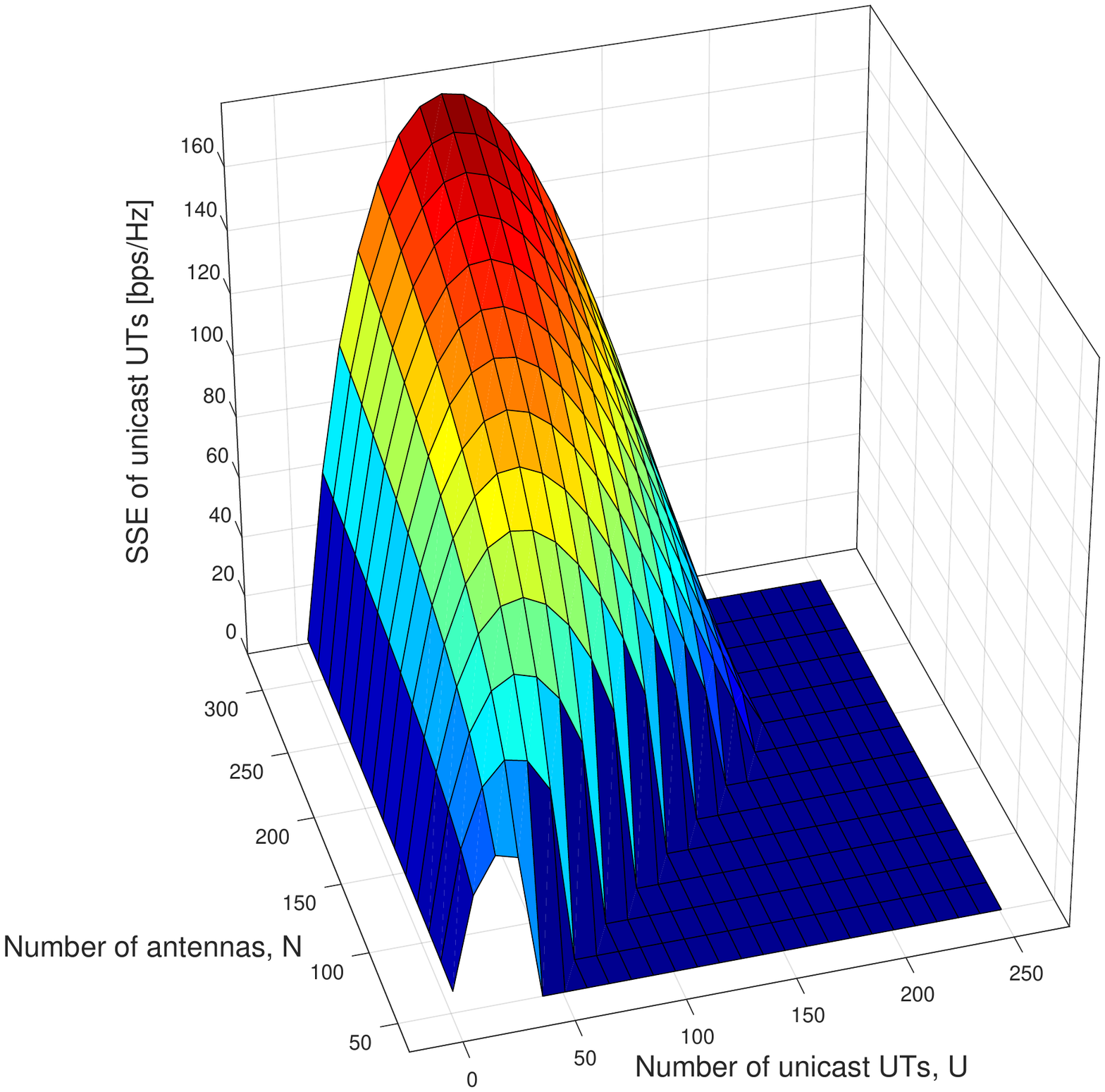}
		\caption{ZF Precoding.}
		\label{FigUnZF}
	\end{subfigure}
\caption{Sum SE (SSE) of unicast UTs versus $U$ and $N$}
\label{FigUn}
\end{figure}

Fig. \ref{FigConvex} presents the Pareto boundary of the MOOP $\mathcal{M}$, achieved based on Theorem \ref{Theo_ParetoBound}, for different number of antennas $N$ at the BS, and MRT and ZF precoding. The horizontal axis is the minimum SE of the multicast UTs and the vertical axis is the SSE of unicast UTs. For each value of $N$, the ratio $P_{un} / P$ is changed from $0$ to $1$ with steps of $P/20$, which are denoted by $21$ marker points. In both Figs. \ref{FigConvex.MRT} and \ref{FigConvex.ZF}, three radial lines are shown, which are presenting specific ratio between $P_{un}$ and $P_{mu}$ ($P_{un}=19P_{mu}$, $P_{un}=P_{mu}$, and $19P_{un}=P_{mu}$). First, the radial lines show adding more antennas improves the objective of both unicast and multicast transmissions. This is due to the improved coherent transmission of signals and it is obtained by employing a large-scale antenna array. Second, recall from Section V that each marker point in Fig. \ref{FigConvex} represents a particular trade-off between the two objectives of the system. Hence Fig. \ref{FigConvex} enables the network designer to allocate the resources according its requirements by selecting the appropriate trade-off. Third, note that the Pareto region for all the considered values of $N$ is convex, as proved in Theorem \ref{convexity}. Fourth, the Pareto boundary does not have any weak Pareto optimal points, as shown in Theorem \ref{Theo_ParetoBound}.

\begin{figure}[]
	\centering
	\begin{subfigure}[b]{0.48\linewidth}
		\centering
		\includegraphics[width=1\columnwidth, trim={2cm 6.5cm 1.8cm 7.2cm},clip]{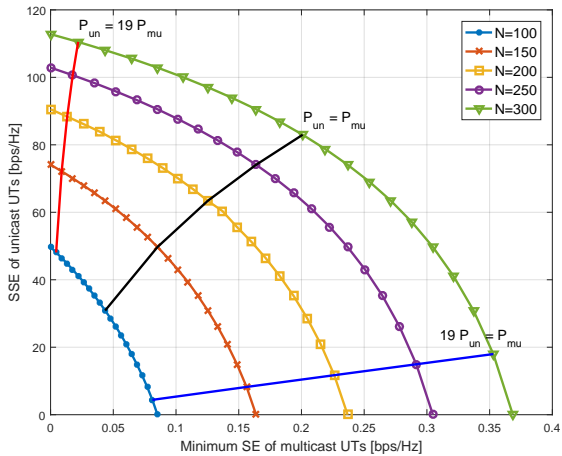} 
		\caption{MRT Precoding.}
		\label{FigConvex.MRT}
	\end{subfigure}
	~
	\begin{subfigure}[b]{0.48\linewidth}
		\centering
		\includegraphics[width=1\columnwidth, trim={2cm 6.5cm 1.8cm 7.2cm},clip]{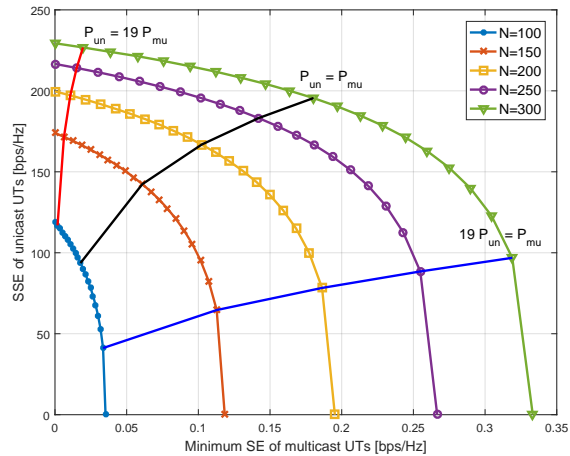}
		\caption{ZF Precoding.}
		\label{FigConvex.ZF}
	\end{subfigure}
\caption{Pareto boundary of MOOP $\mathcal{M}$, for different number of antennas $N$.}
\label{FigConvex}
\end{figure}

\section{Conclusion}
In this paper, we studied the joint unicast and multi-group multicast transmission in massive MIMO systems. Knowing that these systems may have two different and possibly conflicting objectives to fulfill, e.g., MMF for multicast UTs and SSE maximization for unicast UTs, we investigate how to define a joint notion of optimality for them. Therefore, we presented a multiobjective optimization framework to study the conflicting objectives of multicasting and unicasting. We showed that the objective of multicasting cannot be increased without reducing the objective of unicasting, and vice versa. Therefore, we introduced a MOOP and derived its Pareto boundary analytically. Moreover, for any desired point on the Pareto boundary, we determined the value of each of the decision variables such that the desired Pareto boundary point is achieved. We further proved that the Pareto region is convex, thus any point of the Pareto boundary can be obtained by spatial multiplexing of the unicast and multicast UTs. In other words, practical systems can without loss of optimality allow for coexistence between these two use cases. We should note that the multicast transmission in the context of massive MIMO has much potential for further research. For example hybrid precoding, physical layer security, and hardware impairments in the context of massive MIMO multicasting are promising research directions, that have been overlooked so far. We think these aspects would gain more interest in the upcoming years due to the limited existing works and their practicality, i.e., \cite{DemirIcassp,mmWaveMulticast,DaiHPMulcsing,EW18}.

\section{Appendices}



\section*{Appendix A - Proof of Theorem \ref{multicastOpt}}
First we replace the constraint (\ref{MMF}-iv) in $\mathcal{P}1$ with $U\!\!+G \leq \tau \leq \! G $ and treat $\tau$ as a real-valued variable, and name the obtained relaxed problem $\mathcal{A}1$. Note that the solution to $\mathcal{A}1$ is an upper bound on the solution to $\mathcal{P}1$. In the sequel we solve $\mathcal{A}1$ and show that its optimal solution is also a feasible solution of $\mathcal{P}1$, hence it is the optimal solution. Considering $\mathcal{A}1$, now we show that at the optimal solution $P_{mu} = \sum_{j=1}^{G} q_{j}^{dl}$ is equal to $P - P_{un}$. Assume the contrary, then $\forall j$ one can increase $q_{j}^{dl}$ to $\theta  q_{j}^{dl}$ with $\theta = \frac{P-P_{un}}{P_{mu}} > 1$. Therefore the SE of each UT $k$ in group $j$ increases to 
\begin{align}
(1 - \dfrac{\tau}{T}) \log_{2}\left(1+ \dfrac{N q_{j}^{dl} \xi_{jk} }{ \theta^{-1}(1+\eta_{jk} P_{un})  + \eta_{jk}  P_{mu}}\right)
\end{align}
and the optimal solution of $\mathcal{A}1$ improves, which contradicts our optimality assumption. Now let us denote the user with the minimum SE in $j$th multicast group as $kmin_{j}$, i.e., $kmin_{j} = \argmin_{k \in \mathcal{K}_{j}} \mathrm{SE}_{jk,mu}^{\mathrm{MRT}}$. We will prove that at the optimal solution $\forall  i,j \in \mathcal{G}, \mathrm{SE}_{jkmin_{j},mu}^{\mathrm{MRT}} = \mathrm{SE}_{ikmin_{i},mu}^{\mathrm{MRT}}$. Assume the contrary, then $\exists i,j \in \mathcal{G} $ such that $\mathrm{SE}_{jkmin_{j},mu}^{\mathrm{MRT}} > \mathrm{SE}_{ikmin_{i},mu}^{\mathrm{MRT}}$, where $kmin_{i}$ is the user with minimum SE in the system. Now if we respectively change $q_{j}^{dl}$ and $q_{i}^{dl}$ to $q_{j}^{dl} - \delta$ and $q_{i}^{dl} + \delta$ with $0 < \delta < \frac{1+\eta_{jk}P}{N \xi_{jk}} (\mathrm{SINR}_{jkmin_{j},mu}^{\mathrm{MRT}} - \mathrm{SINR}_{ikmin_{i},mu}^{\mathrm{MRT}})$, we can increase $\mathrm{SE}_{ikmin_{i},mu}^{\mathrm{MRT}}$, while ensuring the new value of $\mathrm{SE}_{jkmin_{j},mu}^{\mathrm{MRT}}$ is still bigger than its previous value of $\mathrm{SE}_{ikmin_{i},mu}^{\mathrm{MRT}}$. Repeating this procedure (in case we have multiple UTs with the same minimum SE) we can increase the optimal solution, which contradicts our assumption. Now we prove that $\forall j \in \mathcal{G}, \forall k,s \in \mathcal{K}_{j}, \; \mathrm{SE}_{jk,mu}^{\mathrm{MRT}} = \mathrm{SE}_{js,mu}^{\mathrm{MRT}}$. Assume the contrary then $\exists  k,s \in \mathcal{K}_{j}, \mathrm{SE}_{jk,mu}^{\mathrm{MRT}} > \mathrm{SE}_{js,mu}^{\mathrm{MRT}}$. Denote $x_{jk} = \tau q_{jk}^{up}$. Now we can improve the minimum SE of this group (and the network) by reducing $x_{jk}$ to $x_{jk} - \delta$ with $0 < \delta < (\mathrm{SINR}_{jk,mu}^{\mathrm{MRT}} - \mathrm{SINR}_{js,mu}^{\mathrm{MRT}}) \frac{(1+\sum_{t=1}^{K_{j}} x_{jt} \eta_{jt})(1+\eta_{jk}P)}{N q_{j}^{dl} \eta_{jk}^{2}}$, which is a contradiction. Therefore, at the optimal solution of $\mathcal{A}1$, the SEs are equal, and for multicast group $j$ we have
\begin{align}
\Upsilon_{j} = \dfrac{x_{jk} \eta_{jk}^{2}}{1+\eta_{jk}P} = \dfrac{x_{jt} \eta_{jt}^{2}}{1+\eta_{jt}P} \quad \forall k,t \in \mathcal{K}_{j} \; \forall j \in \mathcal{G}
\end{align} 
where $\Upsilon_{j}$ is a constant. Note that $\mathrm{SE}_{jk,mu}^{\mathrm{MRT}}$ is strictly increasing with $x_{jk}$ and $x_{jk} \leq E_{jk}$, hence its optimal value becomes $x_{jk}^{*} = \frac{1+ \eta_{jk}P}{\eta_{jk}^{2}} \Upsilon_{j}$ with $\Upsilon_{j} = \min_{k \in \mathcal{K}_{j}} \frac{E_{jk} \eta_{jk}^{2}}{1+\eta_{jk}P}$. As for every given $x_{jk}$, the $\mathrm{SE}_{jk,mu}^{\mathrm{MRT}}$ increases monotonically by reducing $\tau$ we have
\begin{align}
\tau^{*} &= U+G \label{tauA1}
\\
q_{jk}^{up*} &= \frac{1+ \eta_{jk}P}{(U+G)\eta_{jk}^{2}} \Upsilon_{j}.
\end{align}
Therefore $\mathrm{SINR}_{jk,mu}^{\mathrm{MRT}*} = \Upsilon_{j} \frac{N q_{j}^{dl*}}{1 + \sum_{t=1}^{K_{j}} x_{jt}^{*} \eta_{jt}}$. Also, as we showed, at the optimal solution all the UTs will have equal SINRs (also equal SEs), i.e., $\Gamma = \mathrm{SINR}_{jk,mu}^{\mathrm{MRT}*} \; \forall j,k$, where $\Gamma$ is a constant. Therefore $q_{j}^{dl*} = \frac{\Gamma}{N \Upsilon_{j}}  (1 + \sum_{t=1}^{K_{j}} x_{jt}^{*} \eta_{jt})$. Summing over all groups and by straightforward operations we have
\begin{align}
\Gamma = \dfrac{P_{mu} N}{P \sum_{j=1}^{G} K_{j} + \sum_{j=1}^{G} \frac{1}{\Upsilon_{j}}  + \sum_{j=1}^{G} \sum_{t=1}^{K_{j}} \frac{1}{\eta_{jt}}}
\end{align}
and the optimal SE follows \eqref{O1}. Note that as \eqref{tauA1} satisfies (\ref{MMF}-iv), the obtained solution is the optimal solution of $\mathcal{P}1$.


\section*{Appendix B - Proof of Theorem \ref{unicastOpt}}

Similar to Appendix A, we replace the constraint (\ref{sumrate}-iv) in $\mathcal{P}2$ with $U\!\!+G \leq \tau \leq \! G $, and name the obtained relaxed optimization problem $\mathcal{A}2$. Note that $\mathcal{A}2$ gives us an upper bound of $\mathcal{P}2$. In the sequel we solve $\mathcal{A}2$ and show that its optimal solution is also a feasible solution of $\mathcal{P}2$, hence it is the optimal solution. Considering $\mathcal{A}2$, we will show that at the optimal solution $P_{un} = \sum_{m=1}^{U} p_{m}^{dl}$ is equal to $P - P_{mu}$.  Assume the contrary, then at the optimal solution $\exists \; \rho>1 : \rho P_{un} = P - P_{mu}$. Now we increase $p_{m}^{dl} \; \forall m$ to $\rho p_{m}^{dl}$, which increases the rate of each UT $m$ to $\log_{2} (1 + \dfrac{N p_{m}^{dl} \vartheta_{m} }{\frac{1+\beta_{m} P_{mu}}{\rho}  + \beta_{m} P_{un} } )$ and hence improves the optimal solution which contradicts our assumption. So at the optimal solution we have $P=P_{mu}+\sum_{m=1}^{U} p_{m}^{dl}$. Now, note that for $\mathcal{A}2$ the objective function is monotonically increasing with respect to $\vartheta_{m}$, which itself is an increasing function of $\tau p_{m}^{up}$. Also for any given value of $\tau p_{m}^{up}$, we can increase the prelog factor $(1-\frac{\tau}{T})$ without changing the SINR, by reducing $\tau$ and setting $p_{m}^{up} = \frac{E_{m}}{\tau}$. So at the optimal solution of $\mathcal{A}2$, $\tau^{*} \!\!= \! U \! + G$ and $p_{m}^{up*} \!= \frac{E_{m}}{U+G}$. Hence we can reduce $\mathcal{A}2$ to the following convex optimization problem
\begin{align}
\label{sum4}
\mathcal{A}3: \max_{ \{ p_{m}^{dl} \} }  & \sum_{m=1}^{U} \alpha_{m} \log_{2} (1 + \dfrac{N p_{m}^{dl} \vartheta_{m}^{*} }{1  + \beta_{m} P } ) 
\\
& s.t. \; \;  \sum_{m=1}^{U} p_{m}^{dl} = P - P_{mu} \tag{\ref{sum4}-i}
\\
& \quad \;\;\; 0 \leq p_{m}^{dl}   \tag{\ref{sum4}-ii}.
\end{align}
where $\vartheta_{m}^{*}=E_{m} \beta_{m}^{2} / (1+E_{m}\beta_{m})$. Now, $\mathcal{A}3$ follows the well-known water-filling structure and its solution is given in \eqref{P3Sol} and the optimal objective value becomes \eqref{O2}. Note that as the obtained solution is also a feasible solution of $\mathcal{P}2$, it is the optimal solution of $\mathcal{P}2$.


\section*{Appendix C - Proof of Theorem \ref{Theo_ParetoBound}}
First note that $\forall (O_{mu}^{*}(P_{un}), O_{un}^{*}(P_{mu})) \in \mathcal{B}_{s}$, as $O_{mu}^{*}(P_{un})$ and $O_{un}^{*}(P_{mu})$ are the optimal solutions of $\mathcal{P}1$ and $\mathcal{P}2$, they are obtained for a feasible point $\mathbf{x}^{*}$ in $\mathcal{X}$, where
\begin{align*}
\mathbf{x}^{*} =  ( \{ q_{j}^{dl*} \}, \{ q_{jk}^{up*} \},  \{ p_{m}^{dl*} \}, \{ p_{m}^{up*} \}, \tau^{*} ).
\end{align*}
Now we prove, every point in $\mathcal{B}_{s}$ is a strong Pareto boundary point, and then we prove, every strong Pareto boundary point is in $\mathcal{B}_{s}$. 

First, we prove that $ \forall (O_{mu}^{*}(P_{un}),O_{un}^{*}(P_{mu})) \in \mathcal{B}_{s}$, it is a strong Pareto boundary point. Assume the contrary, then there exists a $\mathbf{y} \in \mathcal{X}$ such that $ O_{mu}^{*}(P_{un}) < O_{mu}(\mathbf{y}) , O_{un}^{*}(P_{mu})\leq O_{un}(\mathbf{y})$ (the case $ O_{mu}^{*}(P_{un}) \leq O_{mu}(\mathbf{y}) , O_{un}^{*}(P_{mu}) <  O_{un}(\mathbf{y})$ has been omitted for brevity and follows similar procedure). As $O_{mu}^{*}(P_{un})$ is the optimal objective value for the given $P_{un}$, we have $P_{mu} = P - P_{un}$. Based on Theorem \ref{multicastOpt} (or Theorem \ref{multicastOptZF}) and Corollary \ref{Corollarymu}, $ O_{mu}^{*}(P_{un}) < O_{mu}(\mathbf{y})$ implies that $O_{mu}(\mathbf{y})$ is obtained for a higher amount of multicasting power, i.e., $P_{mu} + \delta$ with $0 < \delta$. Therefore the remaining power for unicast transmission in $\mathbf{y}$ is $P - P_{mu} - \delta$. Also the unicast power used by the considered point in $\mathcal{B}_{s}$ is $P_{un} = P - P_{mu}$. Now as $O_{un}^{*}(P_{mu})$ has been allocated more power for unicast transmission than $O_{un}(\mathbf{y})$ and based on Theorem \ref{unicastOpt} (or Theorem \ref{unicastOptZF}) and Corollary \ref{Corollaryun}, $O_{un}^{*}(P_{mu}) > O_{un}(\mathbf{y})$, which contradicts our assumption. Therefore, every point in $\mathcal{B}_{s}$ is a strong Pareto boundary point.

Assume $\mathbf{y} = ( \{ q_{j-y}^{dl} \}, \{ q_{jk-y}^{up} \},  \{ p_{m-y}^{dl} \}, \{ p_{m-y}^{up} \}, \tau_{y} ) \in \mathcal{X}$ and $\mathbf{y} \notin \mathcal{B}_{s}$ is a strong Pareto boundary point, which results in $O_{mu}(\mathbf{y})$ and $O_{un}(\mathbf{y})$ for problems $\mathcal{P}1$ and $\mathcal{P}2$. Denote $P_{mu-y} = \sum_{j=1}^{G} q_{j-y}^{dl}$ and $P_{un-y} = \sum_{m=1}^{U} p_{m-y}^{dl}$. Now select the point $(O_{mu}^{*}(P_{un}), O_{un}^{*}(P_{mu})) \in \mathcal{B}_{s}$ such that $P_{mu} = P_{mu-y}$. Note that such a point exists as $0 \leq P_{mu-y} \leq P $. Then $P_{un} = P - P_{mu-y} \geq P_{un-y}$. Now note that as $O_{mu}^{*}(P_{un})$ is the optimal value of $\mathcal{P}1$, $O_{mu}^{*}(P_{un}) \geq O_{mu}(\mathbf{y})$. Now, if $P_{un} > P_{un-y}$, then based on Theorem \ref{unicastOpt} (or Theorem \ref{unicastOptZF}) and Corollary \ref{Corollaryun}, $O_{un}^{*}(P_{mu}) > O_{un}(\mathbf{y})$, and hence $\mathbf{y}$ is not a Pareto optimal point, which contradicts our assumption. If $P_{un} = P_{un-y}$ and $O_{un}^{*}(P_{mu}) > O_{un}(\mathbf{y})$, again  $\mathbf{y}$ is not a Pareto optimal point and it contradicts our assumption. Finally if $P_{un} = P_{un-y}$ and $O_{un}^{*}(P_{mu}) = O_{un}(\mathbf{y})$, then $\mathbf{y} \in \mathcal{B}_{s}$, which also contradicts our assumption. Therefore, every strong Pareto boundary point is in $\mathcal{B}_{s}$.

\bibliographystyle{IEEEtran}
\bibliography{IEEEabrv,multicastingMassive}

\end{document}